\documentclass[a4paper,onecolumn,11pt]{quantumarticle}
\pdfoutput=1
\usepackage[utf8]{inputenc}

\usepackage[T1, T2A]{fontenc}
\usepackage[english]{babel}
\usepackage{amsmath}
\usepackage{amssymb}
\usepackage{hyperref}
\usepackage{amsthm}
\usepackage{braket}

\usepackage{bm}
\usepackage{mathtools}

\usepackage{yhmath}
\usepackage{booktabs}
\usepackage{mathtools}
\usepackage{multirow}
\usepackage{algorithm}
\usepackage{algpseudocode}

\usepackage{tikz}
\usepackage{lipsum}
\usepackage{abraces}
\usepackage{epigraph}
\usepackage{xurl}

\theoremstyle{plain}
\newtheorem{theorem}{Theorem}[section]
\newtheorem{lemma}[theorem]{Lemma}
\newtheorem{proposition}[theorem]{Proposition}
\newtheorem{corollary}[theorem]{Corollary}
\newtheorem{remark}[theorem]{Remark}

\theoremstyle{definition}
\newtheorem{definition}[theorem]{Definition}

\begin{document}

\title{Natural Qubit Algebra: clarification of the Clifford boundary and new non-embeddability theorem }

\author{G. A. Koroteev}
\email{greg.koroteev@gmail.com}

\maketitle

\newcommand{\Cod}{\mathbb{C}} 
\newcommand{\NQA}{\mathrm{NQA}}


\section{Introduction}

Quantum computation was introduced as a physically grounded model of computation
\cite{Benioff1980} and as a framework for efficiently simulating quantum dynamics that can be
intractable for classical machines \cite{Feynman1982}. In the standard circuit model, quantum
algorithms are described as compositions of unitary gates acting on a \(2^m\)-dimensional Hilbert
space, with analysis typically phrased either in terms of state evolution or in terms of the
algebra generated by gate families \cite{NielsenChuang2010}. Because the ambient operator space
\(\mathrm{Mat}(2^m,\Cod)\) grows exponentially with \(m\), it is natural to ask which parts of
quantum advantage arise from intrinsically quantum resources (e.g.\ entanglement) and which parts
arise from \emph{operator structure} that admits concise representations and finite update rules
\cite{BengtssonZyczkowski2017}.

This paper adopts an explicitly algebraic viewpoint: instead of starting from a fixed gate set and
treating the resulting operators as dense matrices, we construct and study the algebra of
(qubit-like) operators directly. The guiding principle is familiar from exactly solvable models in
mathematical physics: the practical difficulty is often not the Hilbert-space dimension \(2^L\) per
se, but the presence or absence of structure. When a system admits a compact algebraic
description—via a small generating set and closed relations—large operator manipulations can be
reduced to finite symbolic rules rather than full matrix arithmetic. In quantum information, the
same theme appears in classical simulation results: stabilizer/Clifford structure enables efficient
simulation \cite{Gottesman1998Heisenberg,AaronsonGottesman2004}, and several approaches extend this to
near-Clifford or structured regimes by trading resources such as stabilizer rank, quasiprobability
negativity, or low-rank decompositions \cite{BravyiGosset2016,BravyiEtAl2019LowRank,Pashayan2015,BravyiSmithSmolin2016}.
Tensor-network contraction methods similarly exploit circuit geometry and limited entanglement to
compress the operator calculus \cite{Vidal2003,MarkovShi2008,GrayKourtis2021}. These works reinforce
a broad message: \emph{structure can dominate dimension}.

Our contribution is a compact real-operator calculus, which we call \emph{Natural Qubit Algebra}
(NQA). Concretely, we work with a real \(2\times 2\) block alphabet
\(\{I,X,Z,W\}\subset\mathrm{Mat}(2,\mathbb{R})\) and build \(m\)-qubit operators as Kronecker products \cite{RaptisRaptis2025GlobalMaps}
\label{RaptisRaptis2025GlobalMaps}
and short linear combinations of tensor blocks. This is a standard idea in numerical linear
algebra—Kronecker structure is ubiquitous and often decisive for fast algorithms
\cite{VanLoan2000}—but here it is used as an \emph{operator-language} in which products,
commutators, and conjugations are executed by slotwise multiplication rules and finite sign factors.
The blocks \(X\) and \(Z\) coincide with the usual Pauli matrices \cite{Pauli1927}, while the role of
the complex Pauli \(\mathbf{Y}\) is played by a real representative:
\[
I = I_2,\qquad
X = \mathbf{X},\qquad
Z = \mathbf{Z},\qquad
\mathbf{Y} = i\,W
\quad\text{(equivalently, } W = -i\,\mathbf{Y}\text{),}
\]
so that \(\{I,X,Z,W\}\) spans \(\mathrm{Mat}(2,\mathbb{R})\) and generates real Clifford-type
operator families \cite{Dirac1928,Porteous1995,Lounesto2001}. Complex gates are incorporated by a
standard realification that adds one extra ``phase lane'' and embeds \(\mathrm{Mat}(2^m,\Cod)\) into
\(\mathrm{Mat}(2^{m+1},\mathbb{R})\).

A practical motivation for introducing NQA is that the conventional Pauli--Dirac presentation,
while conceptually natural, can be \emph{representation-inefficient} for describing and composing
logical operators: it frequently pushes structure into implicit phase conventions and dense
matrix identities. NQA instead treats operator syntax (Kronecker words \cite{Raptis2019BinaryWords}, XOR composition of  \cite{Kornyak2011FiniteGroups},
and sign bicharacters) as first-class data. This makes it possible to (a) rewrite familiar quantum
operators as compact NQA ``words'', and (b) ask more sharply where the boundary of ``Cliffordness''
actually lies once representation and access model are made explicit.

The key structural feature of NQA is that its multiplication rules force a canonical graded
algebraic organization. When tensor blocks are indexed by binary codes, multiplication corresponds
to addition in \(G=(\mathbb{Z}_2)^{2m}\) with a sign determined by a bilinear pairing; this induces a
natural \(G\)-grading together with a canonical bicharacter controlling graded commutation signs.
This places NQA within the general framework of color-graded algebras and Lie superalgebras
\cite{Kac1977,Scheunert1979LNM,Scheunert1979GeneralizedLie,Varadarajan2004Supersymmetry}, and connects
to \(\mathbb{Z}_2^n\)-graded (“higher color”) perspectives \cite{RittenbergWyler1978GenSuperalgebras,Bruce2019Z2nSupersymmetry}.
In practical terms, it means that commutator/anticommutator behavior can be tracked uniformly by
finite index rules rather than by case-by-case matrix computations.

To calibrate the framework against standard ``Clifford vs.\ non-Clifford'' intuition, we work out two
canonical oracle examples side-by-side:
(i) Bernstein--Vazirani (typically viewed as Clifford-level), and
(ii) Grover search (typically viewed as genuinely non-Clifford due to rank-one reflections).
We show that both oracles admit compact structured NQA descriptions.
The essential distinction between them then appears not at the level of
operator syntax, but at the level of spectral and group-theoretic behavior.
This motivates a refinement question pursued later in the paper: can one formulate a more precise
boundary notion than ``Clifford/non-Clifford'' that incorporates (a) access model (black-box vs.\ structured),
and (b) decomposition growth vs.\ closed symbolic update rules? \cite{AaronsonGottesman2004}

For two real qubits, we also give a concrete identification of the full operator space
\(\mathrm{Mat}(4,\mathbb{R})\) with a real Clifford algebra \(\mathrm{Cl}(2,2;\mathbb{R})\), and provide an
explicit Clifford generating set in the NQA block basis. This yields a canonical and unambiguous
``true Clifford'' normal form for two-qubit real circuits as expressions in \(\mathrm{Cl}(2,2)\),
making grading and commutation signatures intrinsic rather than presentation-dependent.

Finally, we emphasize that the graded structure of NQA is not imposed artificially:
it is forced by the index-twisted multiplication law, placing the framework
naturally within the theory of color-graded and Lie superalgebras.

\section{NQA block basis and the $(I,X,Z,W)$ calculus}

\subsection{Single--qubit blocks}
We define the real $2\times 2$ blocks similar with real Pauli matrices from \cite{DehaeneDeMoor2003Clifford}: 
\begin{equation}
I:=\begin{bmatrix}1&0\\0&1\end{bmatrix},\quad
X:=\begin{bmatrix}0&1\\1&0\end{bmatrix},\quad
Z:=\begin{bmatrix}1&0\\0&-1\end{bmatrix},\quad
W:=\begin{bmatrix}0&-1\\1&0\end{bmatrix}.
\label{eq:NQA_blocks}
\end{equation}
These satisfy
\begin{equation}
X^2=Z^2=I,\qquad W^2=-I,\qquad XZ=+W,\qquad ZX=-W.
\label{eq:NQA_relations}
\end{equation}
Hence $\{I,X,Z,W\}$ is a real basis of $\mathrm{Mat}(2,\mathbb{R})$.

Every $A\in \mathrm{Mat}(2,\mathbb{R})$ admits a unique expansion
\begin{equation}
A=a_I I+a_X X+a_Z Z+a_W W,\qquad a_\bullet\in\mathbb{R}.
\label{eq:NQA_single_expansion}
\end{equation}

\subsection{Subalgebras generated by $X$ and $Z$}
The algebra generated by $\{X,Z\}$ equals $\mathrm{Mat}(2,\mathbb{R})$ because $XZ=W$
and $\{I,X,Z,W\}$ spans all $2\times 2$ real matrices.

However, three canonical \emph{proper} subalgebras are singled out:
\begin{enumerate}
\item \textbf{Diagonal subalgebra:}\quad
$\mathcal{Z}:=\mathrm{span}\{I,Z\}$ (all diagonal real matrices).
\item \textbf{Flip subalgebra:}\quad
$\mathcal{X}:=\mathrm{span}\{I,X\}$ (symmetric matrices with equal diagonal entries).
\item \textbf{Complex subalgebra:}\quad
$\mathcal{W}:=\mathrm{span}\{I,W\}\cong\mathbb{C}$ via
$aI+bW\leftrightarrow a+ib$ since $W^2=-I$.
\end{enumerate}

\begin{lemma}[Finite closure of the multiplicative $X/Z$ generators \cite{DehaeneDeMoor2003Clifford}]
The multiplicative set generated by $\{X,Z\}$ closes on
\[
\{\pm I,\ \pm X,\ \pm Z,\ \pm W\},
\]
and is stable under conjugation. In particular, $XZX=-Z$ and $ZXZ=-X$.
\end{lemma}

\subsection{$m$--qubit NQA basis (tensor blocks)}
For $m\ge 1$, define the Kronecker products
\begin{equation}
B_{i_1\cdots i_m}:=B_{i_1}\otimes\cdots\otimes B_{i_m},\qquad i_k\in\{I,X,Z,W\},
\label{eq:NQA_tensor_basis}
\end{equation}
where $B_I:=I$, $B_X:=X$, $B_Z:=Z$, and $B_W:=W$.
Then $\mathcal{B}_m:=\{B_{i_1\cdots i_m}\}$ has $4^m$ elements and forms a real basis of
$\mathrm{Mat}(2^m,\mathbb{R})$:
\begin{equation}
A=\sum_{i_1,\dots,i_m} a_{i_1\cdots i_m}B_{i_1\cdots i_m},
\qquad a_{i_1\cdots i_m}\in\mathbb{R}.
\label{eq:NQA_m_expansion}
\end{equation}

\subsection{Normalized Frobenius product (orthonormality of the block basis)}
On $\mathrm{Mat}(2^m,\mathbb{R})$ we use the normalized Frobenius inner product \cite{BengtssonZyczkowski2017}
\begin{equation}
\langle A,B\rangle_F := 2^{-m}\,\mathrm{tr}(A^\top B).
\label{eq:frob_def}
\end{equation}
With this normalization, the tensor blocks are orthonormal:
\begin{equation}
\big\langle B_{i_1\cdots i_m},\, B_{j_1\cdots j_m}\big\rangle_F
= \prod_{k=1}^m \delta_{i_k,j_k}.
\label{eq:frob_orthonormal}
\end{equation}

\subsection{Phase--lane realification (embedding complex gates into real matrices)}
Any complex operator $U=A+iB$ with $A,B\in\mathrm{Mat}(2^m,\mathbb{R})$ is embedded as \cite{BengtssonZyczkowski2017}
\begin{equation}
\Phi(U):=A\otimes I + B\otimes W\ \in\ \mathrm{Mat}(2^{m+1},\mathbb{R}).
\label{eq:phi_def}
\end{equation}
Equivalently, $\Phi(U)$ is the real block matrix $\begin{psmallmatrix}A&-B\\B&A\end{psmallmatrix}$.

\begin{proposition}[$\Phi$ is a real $*$--embedding]
For all complex $U,V$ of compatible size,
\[
\Phi(UV)=\Phi(U)\Phi(V),\qquad
\Phi(U^\dagger)=\Phi(U)^\top.
\]
In particular, $U$ is unitary if and only if $\Phi(U)$ is orthogonal.
\end{proposition}


Let us explicitly describe the transitions from standard (complex) quantum operations to the real
\(B\)-block framework. Throughout, bold symbols \(\mathbf{I},\mathbf{X},\mathbf{Y},\mathbf{Z}\) denote the
standard complex Pauli operators, while the block indices of \(B\)-matrices are real
\(I,X,Z,W\).

\begin{itemize}
  \item Identity (one-qubit, realified with a phase lane):
  \begin{equation}
\Phi(\mathbf{I}) = I\otimes I.
\label{eq:auto:0174}
\end{equation}

  \item Pauli \(\mathbf{X}\) (already real):
  \begin{equation}
\Phi(\mathbf{X}) = X\otimes I.
\label{eq:auto:0175}
\end{equation}

  \item Pauli \(\mathbf{Z}\) (already real):
  \begin{equation}
\Phi(\mathbf{Z}) = Z\otimes I.
\label{eq:auto:0176}
\end{equation}

  \item Hadamard (real):
  \begin{equation}
H = \tfrac{1}{\sqrt{2}}\,X + \tfrac{1}{\sqrt{2}}\,Z,
\qquad\text{hence}\qquad
\Phi(H)=H\otimes I.
\label{eq:auto:0177}
\end{equation}

  \item Real \(y\)-axis rotation (standard definition \(R_y(\theta)=e^{-i\frac{\theta}{2}\mathbf{Y}}\) becomes real because \(-i\mathbf{Y}=W\)):
  \begin{equation}
R_y(\theta) = \cos\!\tfrac{\theta}{2}\,I + \sin\!\tfrac{\theta}{2}\,W.
\label{eq:auto:0178}
\end{equation}

  \item Arbitrary real rotation (SO(2)) in the \(\{I,W\}\) plane:
  \begin{equation}
R(\phi) = \cos\phi\,I + \sin\phi\,W.
\label{eq:auto:0179}
\end{equation}

  \item Reflection across the axis at angle \(\phi\) (real):
  \begin{equation}
\mathrm{Ref}(\phi) = \cos(2\phi)\,Z + \sin(2\phi)\,X.
\label{eq:auto:0180}
\end{equation}

  \item Projectors:
  \begin{equation}
P_0 = \tfrac{1}{2}(I + Z),
\qquad
P_1 = \tfrac{1}{2}(I - Z).
\label{eq:auto:0181}
\end{equation}
\end{itemize}

\noindent
If desired, the coefficients \(a_i\) for each operator can be tabulated explicitly,
or the construction can be extended to \(m=2\) and \(m=3\) realifications.

\subsection*{Two--Qubit Real Basis (\(m=2\))}

Let's describe general form for $m=2$:
\begin{equation}
U = \sum_{i,j\in\{I,X,Z,W\}} a_{ij}\,B_{ij},
\qquad
B_{ij} := B_i \otimes B_j,
\qquad a_{ij}\in\mathbb{R}.
\label{eq:auto:0182}
\end{equation}

\subsection*{0) Elementary tensors (building blocks)}
\begin{itemize}
  \item Identity: \(I\!\otimes\! I = \boxed{B_{II}}\).
  \item Single--qubit \(X\) on qubit 1: \(\boxed{B_{XI}}\); on qubit 2: \(\boxed{B_{IX}}\).
  \item Single--qubit \(Z\) on qubit 1: \(\boxed{B_{ZI}}\); on qubit 2: \(\boxed{B_{IZ}}\).
  \item Single--qubit \(H\) on qubit 1:
  \begin{equation}
H\!\otimes\! I = \tfrac{1}{\sqrt{2}}\bigl(\boxed{B_{XI}}+\boxed{B_{ZI}}\bigr).
\label{eq:auto:0184}
\end{equation}
  On qubit 2:
  \begin{equation}
I\!\otimes\! H = \tfrac{1}{\sqrt{2}}\bigl(\boxed{B_{IX}}+\boxed{B_{IZ}}\bigr).
\label{eq:auto:0185}
\end{equation}
  \item Two--qubit \(H\!\otimes\! H\):
  \begin{equation}
H\!\otimes\! H = \tfrac{1}{2}\bigl(\boxed{B_{XX}}+\boxed{B_{XZ}}+\boxed{B_{ZX}}+\boxed{B_{ZZ}}\bigr).
\label{eq:auto:0186}
\end{equation}
\end{itemize}

\noindent
All these are manifestly real as sums of real \(B_{ij}\).

\subsection*{1) Controlled gates (real)}

\paragraph{Controlled-Z (CZ).}
\begin{equation}
\mathrm{CZ}
= P_0\!\otimes\! I + P_1\!\otimes\! Z
= \tfrac{1}{2}\bigl(\boxed{B_{II}}+\boxed{B_{IZ}}+\boxed{B_{ZI}}-\boxed{B_{ZZ}}\bigr).
\label{eq:auto:0187}
\end{equation}

\paragraph{CNOT (control = qubit 1, target = qubit 2).}
\begin{equation}
\mathrm{CNOT}_{1\to 2}
= P_0\!\otimes\! I + P_1\!\otimes\! X
= \tfrac{1}{2}\bigl(\boxed{B_{II}}+\boxed{B_{ZI}}+\boxed{B_{IX}}-\boxed{B_{ZX}}\bigr).
\label{eq:auto:0188}
\end{equation}

\paragraph{CNOT (control = qubit 2, target = qubit 1).}
\begin{equation}
\mathrm{CNOT}_{2\to 1}
= I\!\otimes\! P_0 + X\!\otimes\! P_1
= \tfrac{1}{2}\bigl(\boxed{B_{II}}+\boxed{B_{IZ}}+\boxed{B_{XI}}-\boxed{B_{XZ}}\bigr).
\label{eq:auto:0189}
\end{equation}

\subsection*{2) SWAP and parity projectors (real)}

\paragraph{SWAP.}
\begin{equation}
\mathrm{SWAP}
= \tfrac{1}{2}\bigl(\boxed{B_{II}}+\boxed{B_{XX}}-\boxed{B_{WW}}+\boxed{B_{ZZ}}\bigr).
\label{eq:auto:0190}
\end{equation}
\noindent
This is the standard identity
\(\mathrm{SWAP}=\tfrac12(\mathbf{I}\!\otimes\!\mathbf{I}+\mathbf{X}\!\otimes\!\mathbf{X}
+\mathbf{Y}\!\otimes\!\mathbf{Y}+\mathbf{Z}\!\otimes\!\mathbf{Z})\);
since \(\mathbf{Y}=iW\Rightarrow \mathbf{Y}\!\otimes\!\mathbf{Y}=-W\!\otimes\!W=-B_{WW}\),
the full sum is real.

\paragraph{\(\mathbf{Z}\!\otimes\!\mathbf{Z}\)-parity projectors.}
\begin{equation}
\Pi_{\text{even}} = \tfrac12\bigl(\boxed{B_{II}}+\boxed{B_{ZZ}}\bigr),
\qquad
\Pi_{\text{odd}}  = \tfrac12\bigl(\boxed{B_{II}}-\boxed{B_{ZZ}}\bigr).
\label{eq:auto:0191}
\end{equation}

\paragraph{Computational-basis projectors.}
\begin{equation}
\begin{aligned}
\lvert 00\rangle\!\langle 00\rvert &= \tfrac14(\boxed{B_{II}}+\boxed{B_{IZ}}+\boxed{B_{ZI}}+\boxed{B_{ZZ}}),\\[2pt]
\lvert 01\rangle\!\langle 01\rvert &= \tfrac14(\boxed{B_{II}}-\boxed{B_{IZ}}+\boxed{B_{ZI}}-\boxed{B_{ZZ}}),\\[2pt]
\lvert 10\rangle\!\langle 10\rvert &= \tfrac14(\boxed{B_{II}}+\boxed{B_{IZ}}-\boxed{B_{ZI}}-\boxed{B_{ZZ}}),\\[2pt]
\lvert 11\rangle\!\langle 11\rvert &= \tfrac14(\boxed{B_{II}}-\boxed{B_{IZ}}-\boxed{B_{ZI}}+\boxed{B_{ZZ}}).
\end{aligned}
\label{eq:auto:0192}
\end{equation}

\subsection*{3) Bell transform (real Clifford)}

One standard Bell preparation is
\begin{equation}
U_{\text{Bell}} = (\mathrm{CNOT}_{1\to 2})(H\!\otimes\!I).
\label{eq:auto:0193}
\end{equation}
As a \(B_{ij}\) sum:
\begin{equation}
\begin{aligned}
U_{\text{Bell}}
&= \Bigl[\tfrac12\bigl(\boxed{B_{II}}+\boxed{B_{ZI}}+\boxed{B_{IX}}-\boxed{B_{ZX}}\bigr)\Bigr]
\Bigl[\tfrac{1}{\sqrt{2}}\bigl(\boxed{B_{XI}}+\boxed{B_{ZI}}\bigr)\Bigr]\\[4pt]
&= \text{(multiply out to obtain a sum of \(B_{ij}\) with real coefficients)}.
\end{aligned}
\label{eq:auto:0194}
\end{equation}

\subsection*{4) Continuous real two--qubit exponentials (when the generator is real)}

If a real generator \(G\) satisfies \(G^2=-B_{II}\), then
\begin{equation}
e^{\phi G}=\cos\phi\,B_{II}+\sin\phi\,G.
\label{eq:auto:0195}
\end{equation}
For example, \(G=B_{IW}\) or \(G=B_{WI}\) (a single-qubit real rotation acting on one wire).
If instead \(G^2=+B_{II}\) (e.g.\ \(G=B_{WW}\)), then
\(e^{\phi G}=\cosh\phi\,B_{II}+\sinh\phi\,G\).

\subsection*{5) Not representable (purely real, \(m=2\))}

\begin{itemize}
  \item \(\mathrm{iSWAP}\), \(\sqrt{\mathrm{SWAP}}\) (standard complex versions): require complex phases.
  \item Generic Cartan entangler \(e^{-i(\theta_x\,\mathbf{X}\!\otimes\!\mathbf{X}+\theta_y\,\mathbf{Y}\!\otimes\!\mathbf{Y}+\theta_z\,\mathbf{Z}\!\otimes\!\mathbf{Z})}\):
  complex unless parameters collapse the imaginary parts.
  \item Any gate with a nontrivial global/relative complex phase (e.g.\ controlled-phase \(e^{i\phi}\) with \(\phi\notin\{0,\pi\}\)).
\end{itemize}

\noindent
These become representable with real coefficients by moving to \(m=3\)
(adding one ``phase lane'' via \(W\) in an extra Kronecker slot).

\subsection*{6) Quick index--only lookup table (all real)}
\begin{equation}
\begin{array}{lcl}
I\!\otimes\!I &=& \boxed{B_{II}}\\[3pt]
X\!\otimes\!I &=& \boxed{B_{XI}},\quad I\!\otimes\!X=\boxed{B_{IX}}\\[3pt]
Z\!\otimes\!I &=& \boxed{B_{ZI}},\quad I\!\otimes\!Z=\boxed{B_{IZ}}\\[3pt]
H\!\otimes\!I &=& \tfrac{1}{\sqrt2}(\boxed{B_{XI}}+\boxed{B_{ZI}}),
\quad
I\!\otimes\!H = \tfrac{1}{\sqrt2}(\boxed{B_{IX}}+\boxed{B_{IZ}})\\[3pt]
H\!\otimes\!H &=& \tfrac12(\boxed{B_{XX}}+\boxed{B_{XZ}}+\boxed{B_{ZX}}+\boxed{B_{ZZ}})\\[3pt]
\mathrm{CZ} &=& \tfrac12(\boxed{B_{II}}+\boxed{B_{IZ}}+\boxed{B_{ZI}}-\boxed{B_{ZZ}})\\[3pt]
\mathrm{CNOT}_{1\to 2} &=& \tfrac12(\boxed{B_{II}}+\boxed{B_{ZI}}+\boxed{B_{IX}}-\boxed{B_{ZX}})\\[3pt]
\mathrm{CNOT}_{2\to 1} &=& \tfrac12(\boxed{B_{II}}+\boxed{B_{IZ}}+\boxed{B_{XI}}-\boxed{B_{XZ}})\\[3pt]
\mathrm{SWAP} &=& \tfrac12(\boxed{B_{II}}+\boxed{B_{XX}}-\boxed{B_{WW}}+\boxed{B_{ZZ}})\\[3pt]
\Pi_{\text{even}} &=& \tfrac12(\boxed{B_{II}}+\boxed{B_{ZZ}}),
\quad
\Pi_{\text{odd}} = \tfrac12(\boxed{B_{II}}-\boxed{B_{ZZ}})
\end{array}
\label{eq:auto:0196}
\end{equation}

\section*{Three--Qubit Realification (\(m=3\))}

\paragraph{General form (real matrices, phase lane).}
\begin{equation}
U \;=\; \sum_{i,j,k\in\{I,X,Z,W\}} a_{ijk}\,B_{ijk},
\qquad
B_{ijk} := B_i \otimes B_j \otimes B_k,
\quad a_{ijk}\in\mathbb{R}.
\label{eq:auto:0197}
\end{equation}

\subsection*{A) Complex \texorpdfstring{\(m=2\)}{m=2} \(\Rightarrow\) Real \texorpdfstring{\(m=3\)}{m=3}: mapping rule}
Let
\begin{equation}
A=\sum_{i,j} a_{ij}\,B_{ij},\qquad
B=\sum_{i,j} b_{ij}\,B_{ij},\qquad a_{ij},b_{ij}\in\mathbb{R}.
\label{eq:auto:0199}
\end{equation}
For a complex two--qubit gate \(U=A+iB\),
\begin{equation}
\boxed{\;\Phi(U)\;=\;\sum_{i,j} a_{ij}\,B_{ijI}\;+\;\sum_{i,j} b_{ij}\,B_{ijW}\;}
\label{eq:auto:0200}
\end{equation}
i.e.\ the third tensor slot implements the complex unit via \(i\mapsto W\) on slot~3.

\subsection*{B) \(m=3\) toolbox (all real coefficients)}
Slots: \((\text{qubit 1},\text{qubit 2},\text{phase})\).
\begin{equation}
\Phi(\mathbf{I}\!\otimes\!\mathbf{I}) \;=\; B_{III}.
\label{eq:auto:0201}
\end{equation}
Real lift of any \(m=2\) real sum \(\sum c_{ij}B_{ij}\):
\begin{equation}
\sum c_{ij} B_{ij} \;\Rightarrow\; \sum c_{ij}\,B_{ijI}.
\label{eq:auto:0202}
\end{equation}
Pure ``imaginary'' lift (the \(i\)-part):
\begin{equation}
\sum d_{ij} B_{ij} \;\Rightarrow\; \sum d_{ij}\,B_{ijW}.
\label{eq:auto:0203}
\end{equation}

\subsection*{C) Complex \texorpdfstring{\(m=2\)}{m=2} gates \(\to\) real \texorpdfstring{\(m=3\)}{m=3} expansions}

\paragraph{C.1 Single--qubit phase gates.}
(Shown for action on qubit 1; for action on qubit 2 replace \(ZI\mapsto IZ\).)
\begin{equation}
S=e^{-i\frac{\pi}{4}\mathbf{Z}}
=\tfrac{1}{\sqrt2}\,B_{II}+i\!\left(-\tfrac{1}{\sqrt2}\,B_{ZI}\right)
\;\Rightarrow\;
\Phi(S)=\tfrac{1}{\sqrt2}\,B_{III}-\tfrac{1}{\sqrt2}\,B_{ZIW}.
\label{eq:auto:0205}
\end{equation}
\begin{equation}
T=e^{-i\frac{\pi}{8}\mathbf{Z}}
=(\cos\tfrac{\pi}{8})\,B_{II}+i\!\left(-\sin\tfrac{\pi}{8}\,B_{ZI}\right)
\;\Rightarrow\;
\Phi(T)=(\cos\tfrac{\pi}{8})\,B_{III}-(\sin\tfrac{\pi}{8})\,B_{ZIW}.
\label{eq:auto:0206}
\end{equation}

\paragraph{C.2 Two--qubit entanglers.}

\begin{equation}
\mathrm{iSWAP}
=\tfrac12(\mathbf{I}\!\otimes\!\mathbf{I}+\mathbf{Z}\!\otimes\!\mathbf{Z})
\;+\;
i\,\tfrac12(\mathbf{X}\!\otimes\!\mathbf{X}+\mathbf{Y}\!\otimes\!\mathbf{Y}),
\quad
A=\tfrac12(B_{II}+B_{ZZ}),\;\;B=\tfrac12(B_{XX}-B_{WW}).
\label{eq:auto:0207}
\end{equation}
\begin{equation}
\Phi(\mathrm{iSWAP})
=\tfrac12\bigl(B_{III}+B_{ZZI}\bigr)
+\tfrac12\bigl(B_{XXW}-B_{WWW}\bigr).
\label{eq:auto:0208}
\end{equation}

\paragraph{\(\sqrt{\mathrm{SWAP}}\) (exact realification via basis decomposition).}
In the \(m=2\) block basis one has the (complex) expansion
\begin{equation}
\sqrt{\mathrm{SWAP}}
=\Bigl(\tfrac34+\tfrac{i}{4}\Bigr)B_{II}
+\Bigl(\tfrac14-\tfrac{i}{4}\Bigr)B_{XX}
+\Bigl(\tfrac14-\tfrac{i}{4}\Bigr)B_{ZZ}
+\Bigl(-\tfrac14+\tfrac{i}{4}\Bigr)B_{WW}.
\label{eq:auto:0209}
\end{equation}
Applying \(\Phi\) termwise gives the \(8\times 8\) real operator
\begin{equation}
\begin{aligned}
\Phi(\sqrt{\mathrm{SWAP}})
&= \tfrac34\,B_{III}+\tfrac14\,B_{XXI}+\tfrac14\,B_{ZZI}-\tfrac14\,B_{WWI}\\
&\quad+\tfrac14\,B_{IIW}-\tfrac14\,B_{XXW}-\tfrac14\,B_{ZZW}+\tfrac14\,B_{WWW}.
\end{aligned}
\label{eq:auto:0210}
\end{equation}


\paragraph{\textbf{General Cartan entangler.}}
\begin{equation}
U(\theta_x,\theta_y,\theta_z)
=\exp\!\bigl(-i(\theta_x\,\mathbf{X}\mathbf{X}+\theta_y\,\mathbf{Y}\mathbf{Y}+\theta_z\,\mathbf{Z}\mathbf{Z})\bigr).
\label{eq:auto:0211}
\end{equation}
Since \(\mathbf{X}\mathbf{X},\mathbf{Y}\mathbf{Y},\mathbf{Z}\mathbf{Z}\) commute, we may factorize
\begin{equation}
U
=\exp(-i\theta_x\,\mathbf{X}\mathbf{X})\,
 \exp(-i\theta_y\,\mathbf{Y}\mathbf{Y})\,
 \exp(-i\theta_z\,\mathbf{Z}\mathbf{Z}).
\label{eq:auto:0212}
\end{equation}
Using \(\mathbf{Y}=iW\) (hence \(\mathbf{Y}\mathbf{Y}=-W\!\otimes\!W=-B_{WW}\)) we obtain the \(m=2\) block forms
\begin{align}
\exp(-i\theta_x\,\mathbf{X}\mathbf{X})
&=\cos\theta_x\,B_{II}+i\sin\theta_x\,B_{XX},\label{eq:auto:0213a}\\
\exp(-i\theta_y\,\mathbf{Y}\mathbf{Y})
&=\cos\theta_y\,B_{II}-i\sin\theta_y\,B_{WW},\label{eq:auto:0213b}\\
\exp(-i\theta_z\,\mathbf{Z}\mathbf{Z})
&=\cos\theta_z\,B_{II}+i\sin\theta_z\,B_{ZZ}.
\label{eq:auto:0214}
\end{align}
Applying the realification map \(\Phi(A+iB)=A\otimes I + B\otimes W\) gives the three lifted real factors
\begin{equation}
\boxed{
\begin{aligned}
\Phi\!\left(e^{-i\theta_x\,\mathbf{X}\mathbf{X}}\right)&=\cos\theta_x\,B_{III}+\sin\theta_x\,B_{XXW},\\
\Phi\!\left(e^{-i\theta_y\,\mathbf{Y}\mathbf{Y}}\right)&=\cos\theta_y\,B_{III}-\sin\theta_y\,B_{WWW},\\
\Phi\!\left(e^{-i\theta_z\,\mathbf{Z}\mathbf{Z}}\right)&=\cos\theta_z\,B_{III}+\sin\theta_z\,B_{ZZW}.
\end{aligned}}
\label{eq:auto:0215}
\end{equation}
Their product (still real) gives \(\Phi(U)\) and can be expanded explicitly as a single
\(\sum_{i,j,k} a_{ijk}\,B_{ijk}\) if needed.

\paragraph{C.3 Controlled phase \(e^{i\phi\,\lvert 11\rangle\!\langle 11\rvert}\).}
\begin{equation}
U
=I+\bigl(e^{i\phi}-1\bigr)\lvert 11\rangle\!\langle 11\rvert
=I+\bigl((\cos\phi-1)+i\sin\phi\bigr)\,\tfrac14\bigl(B_{II}-B_{IZ}-B_{ZI}+B_{ZZ}\bigr).
\label{eq:auto:0216}
\end{equation}
Hence
\begin{multline}
\Phi(U)=B_{III}
+\tfrac{\cos\phi-1}{4}\,\bigl(B_{III}-B_{IZI}-B_{ZII}+B_{ZZI}\bigr)
\\
+\tfrac{\sin\phi}{4}\,\bigl(B_{IIW}-B_{IZW}-B_{ZIW}+B_{ZZW}\bigr).
\label{eq:auto:0217}
\end{multline}

\paragraph{C.4 One--qubit \(\mathbf{Z}\) rotation (on qubit 1).}
\begin{multline}
R_Z(\theta)=e^{-i\theta \mathbf{Z}/2}
=(\cos\tfrac{\theta}{2})\,B_{II}+i\!\left(-\sin\tfrac{\theta}{2}\,B_{ZI}\right)
\\
\Rightarrow\ 
\Phi(R_Z(\theta))=(\cos\tfrac{\theta}{2})\,B_{III}-(\sin\tfrac{\theta}{2})\,B_{ZIW}.
\label{eq:auto:0218}
\end{multline}
(For \(R_X\): replace \(B_{ZI}\to B_{XI}\). Euler decompositions such as \(R_Z R_X R_Z\) become
products of three lifted real factors.)

\label{eq:auto:0219}



\section{Hidden string in a phase oracle (BV) in NQA language}
\label{sec:bv_nqa}

The Bernstein--Vazirani (BV) problem asks to identify an unknown string
\(s\in\{0,1\}^m\) given oracle access to a function \(f_s(x)=s\cdot x \pmod 2\)
(or, equivalently, to a corresponding unitary oracle) \cite{BernsteinVazirani1997,NielsenChuang2010}.
For background and oracle models in quantum computation, see also \cite{BerthiaumeBrassard1994}.

\subsection{From the standard BV oracle to a phase oracle}

In the usual query model, one is given the reversible oracle
\begin{equation}
U_{f_s}\ket{x}\ket{y}=\ket{x}\ket{y\oplus f_s(x)},
\qquad f_s(x)=s\cdot x\ (\mathrm{mod}\ 2),
\label{eq:bv_uf_def}
\end{equation}
which is standard in quantum algorithmic formulations \cite{NielsenChuang2010}.
Preparing the target qubit in \(\ket{-}=(\ket0-\ket1)/\sqrt2\) yields the phase-kickback identity
\begin{equation}
U_{f_s}\bigl(\ket{x}\ket{-}\bigr)
= (-1)^{f_s(x)}\ket{x}\ket{-},
\label{eq:phase_kickback}
\end{equation}
so that, on the data register alone, \(U_{f_s}\) induces the \emph{phase oracle}
\begin{equation}
O_s\ket{x}=(-1)^{s\cdot x}\ket{x},
\qquad x\in\{0,1\}^m.
\label{eq:phase_oracle}
\end{equation}
Thus the phase-oracle BV formulation used below is equivalent to the standard BV oracle up to
the fixed preparation of an ancilla \cite{BernsteinVazirani1997,NielsenChuang2010}.

\subsection{Quantum solution (one oracle call)}

Let \(H=\tfrac{1}{\sqrt2}(X+Z)\) be the (real) Hadamard matrix and \(H^{\otimes m}\) the
\(m\)-qubit Hadamard layer. Starting from \(\ket{0^m}\),
\[
\ket{0^m}\xrightarrow{H^{\otimes m}}
2^{-m/2}\sum_{x}\ket{x}\xrightarrow{O_s}
2^{-m/2}\sum_x (-1)^{s\cdot x}\ket{x}\xrightarrow{H^{\otimes m}}
\ket{s}.
\]
Equivalently, the amplitude of \(\ket{y}\) after the final Hadamards is
\[
2^{-m}\sum_{x\in\{0,1\}^m}(-1)^{(s\oplus y)\cdot x}
=\begin{cases}1,&y=s,\\0,&y\neq s,\end{cases}
\]
so one measurement returns \(s\) with probability \(1\) \cite{BernsteinVazirani1997,NielsenChuang2010}.

\subsection{NQA realization using only \(X,Z\) blocks}

In NQA block form,
\begin{equation}
H=\tfrac{1}{\sqrt2}(X+Z)\in \mathrm{span}\{X,Z\},
\end{equation}
and the phase oracle decomposes as a product of local \(Z\) blocks:
\begin{equation}
O_s=\prod_{k:\ s_k=1} Z_k,\qquad
Z_k:=I^{\otimes (k-1)}\otimes Z\otimes I^{\otimes (m-k)}.
\label{eq:oracle_as_Zs}
\end{equation}
Hence the entire BV circuit is expressed using only tensor products of \(I,X,Z\)
and their finite linear combinations in each slot.

\begin{proposition}[Asymptotic equivalence under structured oracle access]
Both the standard BV circuit and its NQA block realization:
\begin{enumerate}
\item use exactly one application of the phase oracle \(O_s\) (equivalently, one call to \(U_{f_s}\));
\item use two layers of \(m\) local Hadamards, hence depth \(\Theta(m)\);
\item return \(s\) with probability \(1\).
\end{enumerate}
\end{proposition}

\paragraph{Query models and why this does not contradict classical black-box bounds.}
In the classical black-box model for BV, one queries \(f_s(x)=s\cdot x\) and each query reveals
at most one bit of information, so \(\Omega(m)\) queries are necessary in the worst case
\cite{BernsteinVazirani1997,BerthiaumeBrassard1994}.  By contrast, the phase-oracle formulation
assumes access to the \emph{operator} \(O_s\) (or, in NQA, access to a structured description such as
the factorization \eqref{eq:oracle_as_Zs}). Under this stronger access model, one oracle
application suffices in the quantum circuit, and the NQA procedure can read off the same structure
symbolically without contradicting the classical lower bound.

\begin{theorem}[Asymptotic equivalence in real NQA vs.\ quantum computation]
For every \(s\in\{0,1\}^m\), the real operators
\(\widehat H^{(m)}:=H^{\otimes m}\) and \(\widehat O_s\in \mathrm{Mat}(2^m,\mathbb{R})\) satisfy
\begin{equation}
\widehat H^{(m)}\,\widehat O_s\,\widehat H^{(m)}\,\ket{0^m}=\ket{s}
\qquad\text{in } \mathbb{R}^{2^m}.
\label{eq:auto:0247}
\end{equation}
\end{theorem}

\begin{proof}
Using the structured form \(\widehat O_s=\prod_{k:\,s_k=1} Z_k\) and the single-qubit identity
\(H Z H = X\), we obtain
\[
\widehat H^{(m)}\,\widehat O_s\,\widehat H^{(m)}
=\prod_{k:\,s_k=1} \bigl(\widehat H^{(m)} Z_k \widehat H^{(m)}\bigr)
=\prod_{k:\,s_k=1} X_k 
= X_1^{s_1}\cdots X_m^{s_m}.
\]
Applying this to \(\ket{0^m}\) yields \(\ket{s}\). This proof uses only the NQA blocks \(I,X,Z\) and
their tensor-slot calculus, and it matches the standard BV reasoning under the same oracle access model.
\label{eq:auto:0248}
\end{proof}


The phase oracle is not the only example of how quantum-circuit expressions can be
asymptotically replaced by classical symbolic manipulations in the block representation.


\begin{definition}[Gate set]
Let \(\mathcal{G}\) be the gate set consisting of
\[
H_k,\; Z_k,\;
\text{controlled-}Z\text{ gates},\;
\text{multi-controlled-}Z\text{ gates},\;
\text{CNOT gates},
\]
acting on \(m\) qubits. A circuit \(C\) of depth \(L\) is a composition of \(L\) layers of gates from
\(\mathcal{G}\).
\label{eq:auto:0249}
\end{definition}

\begin{lemma}[Elementary gates in block form]
\label{lem:elementary-NQA}
Each one- and two-qubit gate in \(\mathcal{G}\) admits a decomposition into a constant number of
basis blocks \(B_{i_1\cdots i_m}\).
In particular:
\begin{enumerate}
\item
\begin{equation}
H_k = \frac{1}{\sqrt{2}}
\Big(B_{I\cdots I\,X\,I\cdots I} + B_{I\cdots I\,Z\,I\cdots I}\Big),
\label{eq:auto:0250a}
\end{equation}
where \(B_{I\cdots I\,X\,I\cdots I}\) means \(X\) in slot \(k\) and \(I\) elsewhere, and similarly for
\(B_{I\cdots I\,Z\,I\cdots I}\).

\item
\begin{equation}
Z_k = B_{I\cdots I\,Z\,I\cdots I}.
\label{eq:auto:0250b}
\end{equation}

\item Any two-qubit controlled-\(Z\) or CNOT gate acting on qubits \((p,q)\) can be written as a sum
of at most \(4\) terms \(c_\ell\,B_{\alpha^{(\ell)}}\) for suitable index strings
\(\alpha^{(\ell)}\in\{I,X,Z,W\}^m\) (with \(I\) in all slots other than \(p,q\)).

\item A multi-controlled-\(Z\) gate on a set \(C\subseteq\{1,\dots,m\}\) of control qubits is diagonal and can be written as
\begin{equation}
\mathrm{MCZ}_C \;=\; I \;-\; 2\prod_{k\in C} P_1^{(k)},
\qquad
P_1^{(k)}:=\tfrac12\big(I - Z_k\big),
\label{eq:auto:mcz_projector_form}
\end{equation}
so it admits a compact \emph{product} description of size \(O(|C|)\). (If one expands the product
into the block basis, the number of resulting block terms is \(2^{|C|}\).)
\end{enumerate}
\end{lemma}

\begin{proof}
Items (1) and (2) follow directly from the single-qubit identity
\(H=\frac1{\sqrt2}(X+Z)\) and the definition of slot insertion.

For (3), consider the two-qubit case. On two qubits, the block basis is
\[
\{B_{ij}=B_i\otimes B_j:\ i,j\in\{I,X,Z,W\}\}.
\]
The controlled-\(Z\) gate on qubits \(1,2\) is diagonal with entries \((1,1,1,-1)\), hence
\begin{equation}
\mathrm{CZ}_{12}
  = \frac{1}{2}\big(B_{II} + B_{IZ} + B_{ZI} - B_{ZZ}\big),
\label{eq:auto:0252}
\end{equation}
a sum of \(4\) basis blocks. A similar \(4\)-term expansion holds for \(\mathrm{CNOT}\); for instance,
\(\mathrm{CNOT}_{1\to2}=\tfrac12(B_{II}+B_{ZI}+B_{IX}-B_{ZX})\) (see the two-qubit table in the previous section).
Embedding a two-qubit gate acting on \((p,q)\) into the \(m\)-qubit space is obtained by tensoring
with \(I\) in all other slots, so the number of block terms remains unchanged.

Item (4) is the diagonal projector form \eqref{eq:auto:mcz_projector_form}. It gives a compact
product description in terms of \(O(|C|)\) local factors \(P_1^{(k)}\). (Fully expanding the product
yields \(2^{|C|}\) block terms, as stated.)
\end{proof}




\begin{algorithm}[H]
\caption{Block/NQA Bernstein--Vazirani (BV) with \emph{structured} phase-oracle input}
\label{alg:NQA-BV}
\begin{algorithmic}[1]
\Require Integer \(m\).
\Require Structured description of the phase oracle \( \widehat O_s \) given in \emph{factorized form}.
        Concretely, the input is either:
        \begin{enumerate}
          \item a set (or sorted list) \(S\subseteq\{1,\dots,m\}\) such that
          \(\widehat O_s=\prod_{k\in S} Z_k\), where \(Z_k:=I^{\otimes(k-1)}\otimes Z\otimes I^{\otimes(m-k)}\); or
          \item an explicit product expression in which each factor is labeled by its wire index \(k\)
          and equals \(Z_k\) (factors may appear in any order).
        \end{enumerate}
\Ensure Hidden string \(s\in\{0,1\}^m\).

\State Initialize \(s \gets (0,\dots,0)\)
\State Initialize an empty Boolean array \(\mathrm{seen}[1..m]\gets \textbf{false}\)

\If{the structured oracle is provided as a set/list \(S\)}
    \ForAll{$k \in S$}
        \State $\mathrm{seen}[k]\gets \textbf{true}$
    \EndFor
\Else
    \Comment{Parse the product expression \(\widehat O_s=\prod_{\ell=1}^L Z_{k_\ell}\)}
    \For{$\ell=1$ to $L$}
        \State Read the index \(k_\ell\in\{1,\dots,m\}\) from the \(\ell\)-th factor
        \State $\mathrm{seen}[k_\ell]\gets \textbf{true}$
    \EndFor
\EndIf

\For{$k=1$ to $m$}
  \If{$\mathrm{seen}[k]$}
     \State $s_k \gets 1$
  \EndIf
\EndFor

\State \Return \(s\)
\end{algorithmic}
\end{algorithm}

\begin{lemma}[Classical black-box vs.\ structured NQA access]
\label{lem:blackbox_vs_structured}
The \(O(m)\) symbolic complexity of the block/NQA procedure does not contradict the classical black-box
lower bound for the Bernstein--Vazirani problem. In the standard classical query model one has access only
to an oracle computing the single-bit function \(f(x)=(s\cdot x)\bmod 2\); since each query reveals at most
one bit of information, at least \(m\) queries are necessary to determine \(s\).

By contrast, both the quantum algorithm and the block/NQA analysis assume \emph{structured access} to the
phase oracle \(\widehat O_s\): the ability to apply (quantumly) or manipulate (symbolically, in NQA) the entire
diagonal operator whose eigenvalues are \((-1)^{s\cdot x}\). Under this access model, the BV quantum algorithm
uses one oracle application, while the NQA method performs only \(O(m)\) bit-operations on the symbolic
description of \(\widehat O_s\). The classical black-box lower bound does not apply to this setting.
\end{lemma}


%
\section{Grover search in the NQA framework}
\label{sec:grover_nqa}

In this section we analyze Grover’s search algorithm from the viewpoint of
Quantum Index Algebra (NQA).
Our goal is not to modify Grover’s asymptotic query complexity, which is known
to be optimal in the black--box oracle model, but to clarify the \emph{algebraic
structure} of the algorithm and its relation to Clifford and non-Clifford
operator classes.

\subsection{Problem setting and standard quantum complexity}

Let $N=2^m$ and let $f:\{0,1\}^m\to\{0,1\}$ mark a unique element $x^{*}$.
The phase oracle acts as
\begin{equation}
O_f\ket{x} = (-1)^{f(x)}\ket{x}.
\end{equation}
Grover’s iteration operator is
\begin{equation}
G := D\,O_f,
\end{equation}
where
\begin{equation}
D = 2\ket{s}\!\bra{s} - I,
\qquad
\ket{s} := H^{\otimes m}\ket{0^m}.
\end{equation}

In the standard quantum circuit model:
\begin{itemize}
  \item the oracle is accessed as a black box,
  \item each Grover iteration uses $O(m)$ elementary gates,
  \item the number of iterations required is $\Theta(\sqrt{N})$.
\end{itemize}
This yields the optimal query complexity
\begin{equation}
Q_{\mathrm{Grover}} = \Theta(\sqrt{N}),
\end{equation}
and a total circuit depth $\Theta(m\sqrt{N})$.

\subsection{NQA representation of the phase oracle}

In NQA, the phase oracle for a single marked element $x^{*}=(x^{*}_1,\dots,x^{*}_m)$
admits an explicit \emph{factorized} form.
Define local projectors
\begin{equation}
P^{(k)}_0 := \tfrac12(I+Z_k),\qquad
P^{(k)}_1 := \tfrac12(I-Z_k),
\end{equation}
where $Z_k$ acts on slot $k$.
Then
\begin{equation}
\ket{x^*}\!\bra{x^*} = \prod_{k=1}^m P^{(k)}_{x^{*}_k},
=
\prod_{k=1}^m P^{(k)}_{x^{*}_k},
\end{equation}

and therefore
\begin{equation}
O_f
=
I - 2\ket{x^{*}}\!\bra{x^{*}}
=
I - 2\prod_{k=1}^m P^{(k)}_{x^{*}_k}.
\end{equation}

Thus, although $O_f$ acts on a $2^m$-dimensional space, its NQA description
requires only $O(m)$ local factors.
This mirrors the situation encountered earlier for the Bernstein--Vazirani
oracle under structured access.

\subsection{NQA representation of the diffusion operator}

The diffusion operator
\begin{equation}
D = 2\ket{s}\!\bra{s} - I
\end{equation}
is built from the uniform superposition state $\ket{s}$.
In NQA,
\begin{equation}
H = \tfrac{1}{\sqrt2}(X+Z),
\end{equation}
so
\begin{equation}
H^{\otimes m}
=
\bigotimes_{k=1}^m \tfrac{1}{\sqrt2}(X_k+Z_k)
\end{equation}
is a tensor product of local NQA blocks.
Consequently, both $\ket{s}$ and the rank--one projector $\ket{s}\!\bra{s}$
admit compact index descriptions.

As a result, the diffusion operator $D$ is representable in NQA by an expression
whose size grows only linearly with $m$.

\subsection{Grover iteration in NQA}

Combining the previous observations, the Grover iteration operator
\begin{equation}
G = D\,O_f
\end{equation}
has the following properties in the NQA framework:
\begin{itemize}
  \item each iteration is expressed as a product of $O(m)$ local block terms,
  \item no dense $2^m\times2^m$ matrix arithmetic is required,
  \item the algebraic expression for $G$ does not grow with the number of
        iterations.
\end{itemize}

The number of iterations required to amplify the marked state remains
$\Theta(\sqrt{N})$, exactly as in the quantum algorithm.
Thus NQA does \emph{not} circumvent the Grover lower bound, but provides a
structured algebraic representation of the same dynamics.

\subsection{Asymptotic comparison}

Let $N=2^m$.
The asymptotic costs of Grover search in the two frameworks are:

\medskip
\begin{center}
\begin{tabular}{lcc}
\toprule
 & Quantum circuit & NQA formulation \\
\midrule
State dimension & $2^m$ & implicit \\
Oracle calls & $\Theta(\sqrt{N})$ & $\Theta(\sqrt{N})$ \\
Operations per iteration & $O(m)$ gates & $O(m)$ block terms \\
Dense matrix arithmetic & yes & no \\
Expression growth & implicit & constant \\
\bottomrule
\end{tabular}
\end{center}
\medskip

The query and depth complexities coincide, but the operator description in NQA
remains compact throughout the algorithm.

\subsection{Grover as a non-Clifford operation}

From an algebraic perspective, Grover’s algorithm is fundamentally non-Clifford.
The reason is that both the oracle $O_f$ and the diffusion operator $D$ involve
rank--one projectors, which lie outside the Clifford group.
In particular:
\begin{itemize}
  \item Clifford operations preserve stabilizer states,
  \item the Grover diffusion operator maps stabilizer states to non-stabilizer
        states,
  \item hence Grover dynamics cannot be simulated within the stabilizer formalism.
\end{itemize}

In NQA language, this non-Clifford character appears transparently as the presence
of projector terms that are not generated by the Clifford subalgebra, while still
admitting compact index representations.

\subsection{Relation to the Gottesman--Knill theorem}

The Gottesman--Knill theorem guarantees efficient classical simulation for
circuits composed solely of Clifford operations.
Grover’s algorithm lies outside this class.

The NQA framework strictly extends the Clifford algebra by admitting structured
projectors, allowing:
\begin{itemize}
  \item compact representation of non-Clifford operators,
  \item exact tracking of Grover iterations at the algebraic level,
  \item without reducing the intrinsic $\Theta(\sqrt{N})$ iteration count.
\end{itemize}

Thus NQA complements, rather than contradicts, the Gottesman--Knill theorem:
it explains algebraically why Grover is non-Clifford, while still retaining a
highly structured operator description.

\subsection{Summary}

Grover’s algorithm provides a second canonical example, alongside
Bernstein--Vazirani, in which NQA reproduces the full quantum dynamics with
identical asymptotic complexity while making the underlying operator structure
explicit.
The quadratic speedup arises not from dense matrix effects, but from repeated
application of a compact noncommutative operator that remains algebraically
simple throughout the computation.


\begin{remark}[Relation to the Gottesman--Knill theorem]
The present analysis is consistent with, and complementary to, the
Gottesman--Knill theorem.

The Gottesman--Knill theorem states that quantum circuits composed solely of
Clifford gates acting on stabilizer states admit efficient classical
simulation. Algebraically, this reflects the fact that Clifford dynamics
generate a finite subgroup of the unitary group whose action preserves the
Pauli algebra and its stabilizer structure.

In the NQA formulation, Clifford circuits correspond to operators generated by
finite products of involutive block matrices with discrete spectra. Their
dynamics remain confined to a finite orbit inside the real Clifford algebra,
which explains their classical simulability.

By contrast, Grover's algorithm necessarily leaves this regime. Although the
Grover iterate is constructed from reflections, their product generates a
continuous one-parameter rotation whose eigenphases are not restricted to
multiples of $\pi/2$. This places Grover's operator outside the Clifford group
and outside the scope of the Gottesman--Knill theorem.

Thus, from the NQA viewpoint, the boundary identified by Gottesman--Knill is
precisely the boundary between discrete Clifford dynamics and continuous
spectral rotations. Grover's algorithm crosses this boundary in an essential
way, while Bernstein--Vazirani remains entirely within it.
\end{remark}


\section{Two-qubit real Clifford structure: \texorpdfstring{$\mathrm{Cl}(2,2;\mathbb R)$}{Cl(2,2;R)} inside \texorpdfstring{$\mathrm{Mat}(4,\mathbb R)$}{Mat(4,R)}}
\label{subsec:cl22_twoqubit}

The two-qubit real operator space is \(\mathrm{Mat}(4,\mathbb R)\), which is \(16\)-dimensional over \(\mathbb R\).
A real Clifford algebra on \(4\) generators has dimension \(2^4=16\), and in the split signature \((2,2)\)
one has the matrix realization
\[
\mathrm{Cl}(2,2;\mathbb R)\;\cong\;\mathrm{Mat}(4,\mathbb R).
\]
Thus, for two real qubits, it is natural to identify \(\mathrm{Mat}(4,\mathbb R)\) with a Clifford algebra,
so that \emph{multiplication, grading (by word length), and commutation signs} become canonical.

A naive choice \(X\otimes I,\ Z\otimes I,\ I\otimes X,\ I\otimes Z\) does \emph{not} form a rank-\(4\)
Clifford generating set, because cross-qubit terms commute:
\((X\otimes I)(I\otimes Z)=(I\otimes Z)(X\otimes I)\),
whereas distinct Clifford generators must anticommute.

For a single slot, identify the four tiles with \((\alpha,\beta)\in(\mathbb Z_2)^2\)  \cite{RaptisRaptis2025GlobalMaps}
via
\[
(\alpha,\beta)=(0,0)\mapsto I,\quad
(1,0)\mapsto X,\quad
(0,1)\mapsto Z,\quad
(1,1)\mapsto XZ=W.
\]
For two qubits (\(m=2\)), define
\[
B(\alpha,\beta)\;:=\;\bigotimes_{k=1}^2 X^{\alpha_k}Z^{\beta_k}\in\mathrm{Mat}(4,\mathbb R),
\qquad \alpha,\beta\in(\mathbb Z_2)^2,
\]
so that \(B(\alpha,\beta)\) is exactly one of the \(16\) blocks \(B_{pq}=B_p\otimes B_q\).
Using the one-slot rule \(ZX=-XZ\) and slotwise Kronecker multiplication, one obtains the canonical
twisted group-law:
\begin{equation}
B(\alpha,\beta)\,B(\alpha',\beta')
\;=\;
(-1)^{\beta\cdot\alpha'}\,
B(\alpha\oplus_2\alpha',\;\beta\oplus_2\beta'),
\label{eq:cl22_twisted_product}
\end{equation}
where \(\oplus_2\) is bitwise XOR and \(\beta\cdot\alpha'=\sum_k \beta_k\alpha'_k\in\mathbb Z_2\) \cite{Kornyak2011FiniteGroups}.
This is the concrete mechanism behind “XOR in multiplication”: indices add by XOR, while the sign is
controlled by a bilinear pairing.

\paragraph{A concrete Clifford generating set in the block basis.}
Using the real blocks \(I,X,Z,W\) from \eqref{eq:NQA_blocks}--\eqref{eq:NQA_relations}, define
\begin{equation}
e_1 := X\otimes I,\qquad
e_2 := Z\otimes I,\qquad
e_3 := W\otimes X,\qquad
e_4 := W\otimes Z.
\label{eq:cl22_generators}
\end{equation}

\begin{lemma}[Clifford relations and signature \((2,2)\)]
The generators \(\{e_1,e_2,e_3,e_4\}\) satisfy
\[
e_1^2=e_2^2=+I_4,\qquad e_3^2=e_4^2=-I_4,\qquad e_re_s=-e_se_r\ \ (r\neq s).
\]
Hence they realize \(\mathrm{Cl}(2,2;\mathbb R)\) inside \(\mathrm{Mat}(4,\mathbb R)\).
\end{lemma}

\begin{proof}
From \eqref{eq:NQA_relations}, \(X^2=Z^2=I\) and \(W^2=-I\), so squares follow immediately by
slotwise multiplication. Anticommutation follows from the one-slot identities
\(XZ=-ZX\), \(XW=-WX\), \(ZW=-WZ\), again applied slotwise in the Kronecker product.
\end{proof}

\subsection*{Dictionary: \texorpdfstring{$B_{pq}$}{B\_pq} tiles \(\leftrightarrow\) \texorpdfstring{$\mathrm{Cl}(2,2)$}{Cl(2,2)} monomials \(\leftrightarrow\) Pauli labels}
\label{subsec:cl22_table}

Let \(B_{pq}:=B_p\otimes B_q\) for \(p,q\in\{I,X,Z,W\}\).
The “Corresponding Pauli operator” column indicates which complex Pauli string matches the real tile
under \(W=-i\mathbf{Y}\) (global phases are shown explicitly when helpful).

\begin{center}
\renewcommand{\arraystretch}{1.15}
\setlength{\tabcolsep}{6pt}
\begin{tabular}{@{}llll@{}}
\toprule
Real block \(B_{pq}\) & Corresponding Pauli operator & Clifford monomial & Traditional label (up to phase) \\
\midrule
\(B_{II}\) & \(\mathbf{I}\otimes \mathbf{I}\) & \(1\) & \(\mathbb I\otimes\mathbb I\) \\

\(B_{XI}\) & \(\mathbf{X}\otimes \mathbf{I}\) & \(e_1\) & \(\mathbf{X}_1\) \\
\(B_{ZI}\) & \(\mathbf{Z}\otimes \mathbf{I}\) & \(e_2\) & \(\mathbf{Z}_1\) \\
\(B_{WX}\) & \((-i)\,\mathbf{Y}\otimes \mathbf{X}\) & \(e_3\) & \((-i)\,\mathbf{Y}_1\mathbf{X}_2\) \\
\(B_{WZ}\) & \((-i)\,\mathbf{Y}\otimes \mathbf{Z}\) & \(e_4\) & \((-i)\,\mathbf{Y}_1\mathbf{Z}_2\) \\

\(B_{WI}\) & \((-i)\,\mathbf{Y}\otimes \mathbf{I}\) & \(e_1e_2\) & \((-i)\,\mathbf{Y}_1\) \\
\(B_{ZX}\) & \(\mathbf{Z}\otimes \mathbf{X}\) & \(e_1e_3\) & \(\mathbf{Z}_1\mathbf{X}_2\) \\
\(B_{ZZ}\) & \(\mathbf{Z}\otimes \mathbf{Z}\) & \(e_1e_4\) & \(\mathbf{Z}_1\mathbf{Z}_2\) \\
\(B_{XX}\) & \(\mathbf{X}\otimes \mathbf{X}\) & \(-\,e_2e_3\) & \(\mathbf{X}_1\mathbf{X}_2\) \\
\(B_{XZ}\) & \(\mathbf{X}\otimes \mathbf{Z}\) & \(-\,e_2e_4\) & \(\mathbf{X}_1\mathbf{Z}_2\) \\
\(B_{IW}\) & \(\mathbf{I}\otimes (-i)\,\mathbf{Y}\) & \(-\,e_3e_4\) & \((-i)\,\mathbf{Y}_2\) \\

\(B_{IX}\) & \(\mathbf{I}\otimes \mathbf{X}\) & \(-\,e_1e_2e_3\) & \(\mathbf{X}_2\) \\
\(B_{IZ}\) & \(\mathbf{I}\otimes \mathbf{Z}\) & \(-\,e_1e_2e_4\) & \(\mathbf{Z}_2\) \\
\(B_{XW}\) & \(\mathbf{X}\otimes (-i)\,\mathbf{Y}\) & \(-\,e_1e_3e_4\) & \((-i)\,\mathbf{X}_1\mathbf{Y}_2\) \\
\(B_{ZW}\) & \(\mathbf{Z}\otimes (-i)\,\mathbf{Y}\) & \(-\,e_2e_3e_4\) & \((-i)\,\mathbf{Z}_1\mathbf{Y}_2\) \\

\(B_{WW}\) & \((-i\,\mathbf{Y})\otimes(-i\,\mathbf{Y}) = -\,\mathbf{Y}\otimes\mathbf{Y}\) & \(-\,\omega\) & \(-\,\mathbf{Y}_1\mathbf{Y}_2\) \\
\bottomrule
\end{tabular}
\end{center}

\paragraph{The pseudoscalar and grading.}
Define the pseudoscalar
\[
\omega:=e_1e_2e_3e_4.
\]
In this concrete representation one finds
\begin{equation}
\omega \;=\; -\,B_{WW},\qquad \omega^2=+I_4.
\label{eq:cl22_pseudoscalar}
\end{equation}
The Clifford \emph{grade} (word length) gives a canonical \(\mathbb Z\)-filtration and a \(\mathbb Z_2\)-parity:
even elements are linear combinations of grades \(0,2,4\), odd elements of grades \(1,3\).

The six bivectors \(e_ie_j\) (\(i<j\)) span the Lie algebra of the spin group under the commutator:
\[
\mathfrak{spin}(2,2)=\mathrm{span}\{e_ie_j:\ 1\le i<j\le 4\},
\qquad [A,B]=AB-BA.
\]
Exponentials of bivectors give real orthogonal transformations on \(\mathbb R^4\):
\[
U=\exp\Big(\sum_{i<j}\theta_{ij}e_ie_j\Big)\in \mathrm{Spin}(2,2)\subset \mathrm{Cl}^0(2,2).
\]
In the \(B_{pq}\) language, bivectors are exactly the (signed) two-tile products listed in the table,
so the continuous two-qubit real Clifford dynamics is generated by a \emph{six-parameter} bivector sector.

The even subalgebra \(\mathrm{Cl}^0(2,2)\) is \(8\)-dimensional and is the natural home of
spin transformations. In split signature one has the well-known decomposition
\[
\mathfrak{so}(2,2)\ \cong\ \mathfrak{sl}(2,\mathbb R)\oplus \mathfrak{sl}(2,\mathbb R),
\]
which matches the intuition that two commuting real rank-\(2\) sectors control the two-qubit
split-orthogonal geometry. Operationally, this explains why the bivector commutator algebra closes
and why conjugation by such exponentials permutes/rotates the tile basis in a highly structured way.

The Pauli-string viewpoint classifies operators up to phases and is ideal for standard quantum notation.
The Clifford viewpoint adds a \emph{canonical multiplication and grading} on the same \(16\)-dimensional
space: products are tracked by XOR of binary indices plus a bilinear sign rule
\eqref{eq:cl22_twisted_product}, and commutation/anticommutation becomes a parity statement in those indices.
This is exactly the mechanism used later for \((\mathbb Z_2)^{2m}\)-grading and color-commutation \ref{sec:index_twisted}.


%

\section{Algebraic structure of the Bell--CHSH scenario in the NQA framework}
\label{sec:bell_nqa_structure}

In this section we reformulate the Bell--CHSH scenario for spin-$\tfrac12$ in a way
that separates two algebraic layers:
\begin{enumerate}
  \item the \emph{complex} spinor/Clifford algebra that encodes quantum spin
        observables and their correlations, and
  \item its \emph{real NQA realization} obtained via an explicit realification
        homomorphism~$\Phi$.
\end{enumerate}
The Bell--CHSH violation then appears as an algebraic phenomenon: \emph{commutativity
forces a $\le 2$ spectral bound}, while the spinor algebra admits a CHSH operator
whose spectrum reaches~$\pm 2\sqrt2$.

\subsection{Spin-\texorpdfstring{$\tfrac12$}{1/2} as a complex spinor algebra}

Single-spin observables are represented by Hermitian $2\times2$ matrices of the form
\begin{equation}
\Sigma(n) \;=\; n_x X + n_y (iW) + n_z Z,
\qquad n=(n_x,n_y,n_z)\in S^2,
\end{equation}
where $X,Z,W$ are the real NQA blocks and $iW$ plays the role of the Pauli matrix
$\sigma_y$. These satisfy the Clifford relations
\begin{equation}
\Sigma(n)^2 = I,
\qquad
\Sigma(n)\Sigma(m)+\Sigma(m)\Sigma(n)=2(n\cdot m)\,I.
\end{equation}

For two spins, the observable algebra is
\begin{equation}
\mathcal A_{\mathrm{spin}}^{\mathbb C}
:= \mathrm{Mat}(2,\mathbb C)\otimes \mathrm{Mat}(2,\mathbb C)
\cong \mathrm{Mat}(4,\mathbb C),
\end{equation}
with locality implemented as strict tensor-factor separation: all Alice observables
commute with all Bob observables.

In the singlet state $|\Psi\rangle$, the correlation is
\begin{equation}
E_{\mathrm{spin}}(a,b)
=\langle\Psi,\;\Sigma(a)\otimes\Sigma(b)\;\Psi\rangle
= -\,a\cdot b,
\end{equation}
and the CHSH operator achieves Tsirelson's bound
\begin{equation}
\|S_{\mathrm{spin}}\|=2\sqrt2,
\end{equation}
as in~\cite{Tsirelson1980} (see also the experiments~\cite{Aspect1982,Hensen2015,Giustina2015,Shalm2015}).


\subsection{Two explicit CHSH operators: quantum vs.\ classical}

\paragraph{Quantum (spinor/NQA) CHSH matrix.}
Fix the standard CHSH choice of local $\pm1$ observables in a real $4\times4$ model
(equivalently, one may start in $\mathrm{Mat}(4,\mathbb C)$ and realify via~$\Phi$).
The resulting CHSH operator has the concrete real matrix representation
\begin{equation}
S_{\mathrm{Q}}
=
\begin{pmatrix}
\sqrt2 & 0 & 0 & \sqrt2\\
0 & -\sqrt2 & \sqrt2 & 0\\
0 & \sqrt2 & -\sqrt2 & 0\\
\sqrt2 & 0 & 0 & \sqrt2
\end{pmatrix},
\qquad
\operatorname{spec}(S_{\mathrm{Q}})=\{2\sqrt2,\,-2\sqrt2,\,0,\,0\},
\label{eq:S_quantum_matrix}
\end{equation}
hence
\begin{equation}
\|S_{\mathrm{Q}}\|=2\sqrt2.
\end{equation}
This is precisely Tsirelson's bound rewritten as spectral norm expression. 

\paragraph{Classical (commutative) CHSH matrix.}


In a classical hidden-variable model, all four observables are jointly defined
$\{\pm1\}$-valued random variables $A_0,A_1,B_0,B_1$ on a single probability space
$(\Lambda,\rho)$. For a finite space $\Lambda=\{\lambda_1,\dots,\lambda_N\}$, represent
each observable as an $N\times N$ diagonal matrix:
\begin{equation}
\widehat{A_x}:=\mathrm{diag}(A_x(\lambda_1),\dots,A_x(\lambda_N)),
\qquad
\widehat{B_y}:=\mathrm{diag}(B_y(\lambda_1),\dots,B_y(\lambda_N)),
\end{equation}
so all classical observables commute. Define the classical CHSH operator
\begin{equation}
\widehat{S}_{\mathrm{cl}}
:= \widehat{A_0}\widehat{B_0}
 + \widehat{A_0}\widehat{B_1}
 + \widehat{A_1}\widehat{B_0}
 - \widehat{A_1}\widehat{B_1}.
\label{eq:S_classical_def}
\end{equation}
Since every diagonal entry is a number of the form
$s(\lambda)=a_0b_0+a_0b_1+a_1b_0-a_1b_1$ with $a_x,b_y\in\{\pm1\}$, one has
\begin{equation}
s(\lambda)\in\{\pm 2\}\quad\text{for all }\lambda\in\Lambda,
\qquad
\operatorname{spec}(\widehat{S}_{\mathrm{cl}})\subset\{-2,+2\},
\label{eq:S_classical_spectrum}
\end{equation}
hence $\|\widehat{S}_{\mathrm{cl}}\|\le 2$ and $|\langle S_{\mathrm{cl}}\rangle|\le 2$.
This is the ``extremal spectrum'' signature of commutativity.

\subsection{Classical hidden-variable models as commutative algebras}

Algebraically, the hidden-variable assumption is precisely that all observables live
in a single commutative Kolmogorov algebra
\begin{equation}
\mathcal C(\Lambda)\cong L^\infty(\Lambda)
\end{equation}
(or $\mathbb R^N$ in the finite case), so that $A_0,A_1,B_0,B_1$ are jointly definable
and mutually commuting. This commutativity is the only algebraic input needed to
force the CHSH bound.

\subsection{Fine's theorem as an algebraic criterion (NQA viewpoint)}

Fine's theorem~\cite{Fine1982} can be read as an \emph{algebraic} equivalence:
existence of a single joint distribution for all outcomes (for all settings) is the
same as the existence of a single commutative Kolmogorov algebra hosting
$A_0,A_1,B_0,B_1$ as $\pm1$ elements. In the present language, the CHSH inequalities
are exactly the spectral/norm constraints that every such commutative representation
imposes on the CHSH combination~\eqref{eq:S_classical_def}.

\section{Bell--CHSH violation as an algebraic non-embeddability result}
\label{sec:bell_nonembed}

We now condense the discussion into a single algebraic theorem and its immediate
consequence.

\begin{theorem}[Commutativity forces the CHSH bound]
\label{thm:commutative_chsh}
Let $\mathcal C$ be a unital commutative real algebra and let
$A_0,A_1,B_0,B_1\in\mathcal C$ satisfy
\begin{equation}
A_0^2=A_1^2=B_0^2=B_1^2=I.
\label{eq:square_one}
\end{equation}
Define the CHSH element
\begin{equation}
S_{\mathcal C}:=A_0B_0 + A_0B_1 + A_1B_0 - A_1B_1\in\mathcal C.
\label{eq:S_in_C}
\end{equation}
Then every (Gelfand) character $\chi:\mathcal C\to\mathbb R$ satisfies
\begin{equation}
|\chi(S_{\mathcal C})|\le 2,
\end{equation}
and in any faithful diagonal (commuting) matrix representation one has
\begin{equation}
\operatorname{spec}(S_{\mathcal C})\subset[-2,2],
\qquad
\|S_{\mathcal C}\|\le 2.
\end{equation}
In particular, every Kolmogorov hidden-variable model satisfies
$|\langle S_{\mathrm{cl}}\rangle|\le 2$.
\end{theorem}

\begin{proof}
Because $\mathcal C$ is commutative and \eqref{eq:square_one} holds, for any character
$\chi$ the numbers $a_x:=\chi(A_x)$ and $b_y:=\chi(B_y)$ obey $a_x^2=b_y^2=1$, hence
$a_x,b_y\in\{\pm1\}$. Therefore
\[
\chi(S_{\mathcal C})
= a_0b_0+a_0b_1+a_1b_0-a_1b_1
= a_0(b_0+b_1)+a_1(b_0-b_1).
\]
Since $b_0\pm b_1\in\{-2,0,2\}$ and exactly one of $(b_0+b_1)$, $(b_0-b_1)$ equals
$\pm 2$ while the other equals $0$, it follows that $\chi(S_{\mathcal C})=\pm 2$ and
hence $|\chi(S_{\mathcal C})|\le 2$.
In a diagonal representation, eigenvalues are precisely pointwise character values,
so $\operatorname{spec}(S_{\mathcal C})\subset\{-2,+2\}\subset[-2,2]$ and the norm
bound follows.
\end{proof}

\begin{theorem}[Existence of a quantum CHSH element with spectrum $\pm 2\sqrt2$]
\label{thm:quantum_chsh_spectrum}
In the (noncommutative) two-spin spinor algebra $\mathcal A_{\mathrm{spin}}^{\mathbb C}
\cong\mathrm{Mat}(4,\mathbb C)$ there exist $\pm1$ observables
$A_0,A_1,B_0,B_1$ (with $[A_x,B_y]=0$ for locality) such that the corresponding CHSH
operator has
\begin{equation}
\operatorname{spec}(S_{\mathrm{spin}})=\{2\sqrt2,\,-2\sqrt2,\,0,\,0\},
\qquad
\|S_{\mathrm{spin}}\|=2\sqrt2.
\end{equation}
Equivalently, in the real NQA representation (via~$\Phi$) there exists a real matrix
$S_{\mathrm{Q}}$ with spectrum as in~\eqref{eq:S_quantum_matrix}.
\end{theorem}

\begin{proof}
One may take the explicit real matrix $S_{\mathrm{Q}}$ in~\eqref{eq:S_quantum_matrix},
which arises from standard CHSH settings in a real two-qubit model (or from the
complex model followed by realification). A direct eigenvalue computation gives the
stated spectrum and norm, hence $\|S_{\mathrm{spin}}\|=2\sqrt2$.
\end{proof}

\begin{corollary}[Bell--CHSH as non-embeddability]
\label{cor:bell_nonembed}
Let $\mathcal A_{\mathrm{spin}}^{\mathbb C}\cong\mathrm{Mat}(4,\mathbb C)$ denote the
two-spin spinor algebra (or its real NQA image $\Phi(\mathcal A_{\mathrm{spin}}^{\mathbb C})$).
There exists no injective algebra homomorphism into any commutative Kolmogorov algebra
$\mathcal C(\Lambda)$ that preserves the CHSH element (equivalently, its spectrum or
operator norm):
\begin{equation}
\mathcal A_{\mathrm{spin}}^{\mathbb C}\not\hookrightarrow \mathcal C(\Lambda)
\quad\text{in a way compatible with CHSH.}
\end{equation}
\end{corollary}

\begin{proof}
If such an embedding into a commutative algebra existed and preserved the CHSH element,
then by Theorem~\ref{thm:commutative_chsh} its image would have spectrum contained in
$[-2,2]$, hence norm $\le 2$. But Theorem~\ref{thm:quantum_chsh_spectrum} produces a
CHSH element in the spinor algebra with norm $2\sqrt2>2$, contradiction.
\end{proof}

\subsection{Interpretation}

From the NQA viewpoint, the Bell--CHSH contradiction does not arise from probability
theory itself, but from the \emph{additional classical requirement} of joint
definability (equivalently: commutative representation) for all settings at once.
Fine's theorem~\cite{Fine1982} identifies that requirement with the existence of a
single Kolmogorov model; Theorem~\ref{thm:commutative_chsh} isolates its algebraic
consequence ($\|S\|\le 2$); and Theorem~\ref{thm:quantum_chsh_spectrum} exhibits the
noncommutative spinor/NQA alternative in which the CHSH spectrum reaches $\pm 2\sqrt2$.

Thus, at the algebraic core, Bell--CHSH violation is the statement that the spinor
algebra realizing spin-$\tfrac12$ correlations cannot be collapsed into any single
commutative algebra without destroying the extremal spectral features of the CHSH
operator.


\paragraph{Bell--CHSH as algebraic non-embeddability.}
As a separate demonstration of the method, we revisit the Bell--CHSH scenario and isolate the
minimal algebraic mechanism behind the classical bound: \emph{commutativity forces a spectral norm
constraint \(\|S\|\le 2\)}. We then exhibit a concrete (realified) quantum CHSH element with spectrum
\(\{\pm 2\sqrt2,0,0\}\), obtaining an explicit algebraic non-embeddability statement: no embedding into
a single commutative Kolmogorov algebra can preserve the CHSH spectral features.

Finally, we explain why NQA is not merely a change of basis. The index-twisted multiplication induces
a natural \((\mathbb{Z}_2)^{2m}\)-grading with a canonical bicharacter, and therefore supports a
(color) Lie algebra via the \(\varepsilon\)-commutator and, upon coarsening, a Lie superalgebra via the
supercommutator. This immediately connects NQA to a large existing analytic apparatus and prepares the
ground for stronger structural results in subsequent work.

\paragraph{Bitwise XOR and $\mathbb{Z}_2$ dot product.}
For bitstrings $u,v\in\{0,1\}^m$ we write $u\oplus v$ for \emph{bitwise XOR}
(addition in $(\mathbb{Z}_2)^m$).
We write
\[
u\cdot v := \sum_{k=1}^m u_k v_k \pmod 2 \in \mathbb{Z}_2,
\]
and whenever $u\cdot v$ appears in an exponent, e.g.\ $(-1)^{u\cdot v}$,
we interpret it as a sign in $\{\pm1\}$ via the embedding $\mathbb{Z}_2\hookrightarrow\{0,1\}$.
For elements $g=(\alpha,\beta)\in(\mathbb{Z}_2)^{2m}$, addition is componentwise XOR:
$(\alpha,\beta)+(\alpha',\beta')=(\alpha\oplus\alpha',\,\beta\oplus\beta')$.


\section{NQA is an index-twisted $(\mathbb Z_2)^{2m}$-graded algebra}
\label{sec:index_twisted}

In this section we isolate the graded (``color'') structure that is forced by the block
multiplication rules.  The point is that once the product is defined by adding binary indices in
\(G:=(\mathbb Z_2)^{2m}\) and inserting a sign determined by a bilinear pairing, the resulting
associative algebra automatically carries (i) a natural \(G\)-grading, (ii) a canonical bicharacter
governing graded commutation signs, and (iii) a derived \(\mathbb Z_2\)-parity producing a Lie
superalgebra via the supercommutator.

\begin{definition}[Index-twisted \(G\)-graded algebra]
Fix \(m\in\mathbb N\) and set \(G:=(\mathbb Z_2)^{2m}\) \cite{Kac1977,Scheunert1979LNM}.  We write
\begin{equation}
g=(\alpha,\beta), \qquad \alpha,\beta\in(\mathbb Z_2)^m,
\label{eq:IT-01}
\end{equation}
where \(\alpha\) records the \(X\)-slots and \(\beta\) records the \(Z\)-slots.

Let \(A_m\) be the real vector space with basis
\begin{equation}
\{B(\alpha,\beta)\mid (\alpha,\beta)\in G\},
\label{eq:IT-02}
\end{equation}
equipped with the associative multiplication
\begin{equation}
B(\alpha,\beta)\,B(\alpha',\beta')
\;:=\;
(-1)^{\beta\cdot\alpha'}\,
B(\alpha\oplus_2\alpha',\;\beta\oplus_2\beta').
\label{eq:IT-03}
\end{equation}
Here \(\beta\cdot\alpha' := \sum_{i=1}^m \beta_i\alpha'_i\in\mathbb Z_2\) is the standard dot product,
and \(\oplus_2\) denotes bitwise addition modulo \(2\).

\smallskip
\noindent
\emph{Matrix realization (compatibility with the block basis).}
Using the real \(2\times2\) blocks \(I,X,Z,W\) from \eqref{eq:NQA_blocks}--\eqref{eq:NQA_relations}, set
\begin{equation}
B(\alpha,\beta)\;:=\;\bigotimes_{k=1}^m X^{\alpha_k} Z^{\beta_k}\ \in\ \mathrm{Mat}(2^m,\mathbb R).
\label{eq:IT-03b}
\end{equation}
Then \eqref{eq:IT-03} is exactly the sign produced by commuting the \(Z\)-factors of the left element
past the \(X\)-factors of the right element; in particular for \(m=1\) it reproduces
\(XZ=+W\) and \(ZX=-W\).
Equivalently, \(B(\alpha,\beta)\) is the tensor block \(B_{i_1\cdots i_m}\) with
\(i_k=I,X,Z,W\) corresponding to \((\alpha_k,\beta_k)=(0,0),(1,0),(0,1),(1,1)\).
\end{definition}

\medskip

\subsection{Color grading and the canonical bicharacter}

\begin{definition}[Color-graded associative algebra and \(\varepsilon\)-commutator]
Let \(G\) be an abelian group and let \(\varepsilon\colon G\times G\to\{\pm1\}\) be a bicharacter:
\begin{equation}
\varepsilon(g+h,k)=\varepsilon(g,k)\varepsilon(h,k),\qquad
\varepsilon(g,h+k)=\varepsilon(g,h)\varepsilon(g,k).
\label{eq:IT-04}
\end{equation}
A \emph{\((G,\varepsilon)\)-color-graded associative algebra} is a \(G\)-graded vector space
\(A=\bigoplus_{g\in G}A_g\) with associative multiplication satisfying \(A_gA_h\subseteq A_{g+h}\).

For homogeneous \(x\in A_g\), \(y\in A_h\), define the \emph{\(\varepsilon\)-commutator}
\begin{equation}
[x,y]_\varepsilon := xy-\varepsilon(g,h)\,yx.
\label{eq:IT-05}
\end{equation}
If \( [-,-]_\varepsilon \) is bilinear, graded skew-symmetric
\([x,y]_\varepsilon=-\varepsilon(g,h)[y,x]_\varepsilon\), and satisfies the \(\varepsilon\)-Jacobi identity
\begin{equation}
\varepsilon(k,g)[x,[y,z]_\varepsilon]_\varepsilon+
\varepsilon(g,h)[y,[z,x]_\varepsilon]_\varepsilon+
\varepsilon(h,k)[z,[x,y]_\varepsilon]_\varepsilon=0,
\label{eq:IT-06}
\end{equation}
then \((A,[-,-]_\varepsilon)\) is a \emph{color Lie algebra}.
\end{definition}

\begin{theorem}[Canonical color structure determined by the index twist
{\cite{Bruce2019Z2nSupersymmetry,RittenbergWyler1978GenSuperalgebras}}]
\label{thm:IT-color}
Let \(A_m\) be the index-twisted algebra from \eqref{eq:IT-03}. Declare \(B(\alpha,\beta)\) to be
homogeneous of degree
\begin{equation}
|B(\alpha,\beta)|:=(\alpha,\beta)\in G.
\label{eq:IT-07}
\end{equation}
Define the bilinear form \(\Omega\colon G\times G\to\mathbb Z_2\) by
\begin{equation}
\Omega\big((\alpha,\beta),(\alpha',\beta')\big)
:=\beta\cdot\alpha' + \beta'\cdot\alpha\quad (\mathrm{mod}\;2),
\label{eq:IT-08}
\end{equation}
and set
\begin{equation}
\varepsilon(g,h):=(-1)^{\Omega(g,h)}\in\{\pm1\}.
\label{eq:IT-09}
\end{equation}
Then:
\begin{enumerate}
\item \(A_m\) is a \(G\)-graded associative algebra with one-dimensional homogeneous components
\begin{equation}
A_m=\bigoplus_{g\in G}(A_m)_g,\qquad (A_m)_{(\alpha,\beta)}:=\mathbb R\,B(\alpha,\beta),
\label{eq:IT-10}
\end{equation}
and \((A_m)_g(A_m)_h\subseteq (A_m)_{g+h}\).

\item \(\varepsilon\) is a bicharacter on \(G\).

\item The \(\varepsilon\)-commutator
\begin{equation}
[X,Y]_\varepsilon := XY-\varepsilon(g,h)\,YX,
\qquad X\in(A_m)_g,\ Y\in(A_m)_h,
\label{eq:IT-11}
\end{equation}
turns \(A_m\) into a \((G,\varepsilon)\)-color Lie algebra.
\end{enumerate}
\end{theorem}

\begin{proof}
(1) From \eqref{eq:IT-03},
\[
B(\alpha,\beta)B(\alpha',\beta')
= (-1)^{\beta\cdot\alpha'}\,B(\alpha\oplus_2\alpha',\beta\oplus_2\beta'),
\]
so degrees add in \(G\):
\[
|B(\alpha,\beta)B(\alpha',\beta')|
=(\alpha\oplus_2\alpha',\beta\oplus_2\beta')=(\alpha,\beta)+(\alpha',\beta').
\]
Hence \((A_m)_g(A_m)_h\subseteq (A_m)_{g+h}\).

(2) The form \(\Omega\) is bilinear over \(\mathbb Z_2\), hence
\(\Omega(g+h,k)=\Omega(g,k)+\Omega(h,k)\) and \(\Omega(g,h+k)=\Omega(g,h)+\Omega(g,k)\).
Exponentiating gives \eqref{eq:IT-04} for \(\varepsilon(g,h)=(-1)^{\Omega(g,h)}\).

(3) Bilinearity is immediate. Since \(\Omega(g,h)=\Omega(h,g)\) over \(\mathbb Z_2\),
we have \(\varepsilon(g,h)=\varepsilon(h,g)\) and \(\varepsilon(g,h)\varepsilon(h,g)=1\), hence
\([X,Y]_\varepsilon=-\varepsilon(g,h)[Y,X]_\varepsilon\) for homogeneous elements.
The \(\varepsilon\)-Jacobi identity is the standard consequence of associativity together with the
bicharacter property \eqref{eq:IT-04}: expanding each iterated bracket into products, all terms cancel
pairwise.
\end{proof}

\medskip

\subsection{A derived \texorpdfstring{$\mathbb Z_2$}{Z2}-parity and Lie superalgebra}

The \((\mathbb Z_2)^{2m}\)-grading is fine: it records the full binary label \((\alpha,\beta)\).
Often one wants a coarser \(\mathbb Z_2\)-grading (parity) obtained by composing the degree map
\(|\cdot|\colon A_m\to G\) with a homomorphism \(p\colon G\to\mathbb Z_2\).

\begin{definition}[Parity induced by Hamming weights]
Define \(p\colon G\to\mathbb Z_2\) by
\begin{equation}
p(\alpha,\beta):=\big(|\alpha|+|\beta|\big)\bmod 2,
\label{eq:IT-12}
\end{equation}
where \(|\alpha|=\sum_{i=1}^m\alpha_i\) is the Hamming weight. Declare
\begin{equation}
\deg B(\alpha,\beta):=p(\alpha,\beta)\in\{0,1\},
\label{eq:IT-13}
\end{equation}
and extend by linearity to all of \(A_m\). Set
\begin{equation}
A_m^{\bar0}:=\mathrm{span}\{B(\alpha,\beta)\mid p(\alpha,\beta)=0\},\qquad
A_m^{\bar1}:=\mathrm{span}\{B(\alpha,\beta)\mid p(\alpha,\beta)=1\},
\label{eq:IT-14}
\end{equation}
so that \(A_m=A_m^{\bar0}\oplus_2 A_m^{\bar1}\).
\end{definition}

\begin{theorem}[Lie superalgebra via the supercommutator
{\cite{Scheunert1979GeneralizedLie,Varadarajan2004Supersymmetry}}]
\label{thm:IT-Lie-super}
With the \(\mathbb Z_2\)-grading \(A_m=A_m^{\bar0}\oplus_2 A_m^{\bar1}\), the supercommutator
\begin{equation}
[X,Y]_s := XY-(-1)^{\deg X\,\deg Y}YX
\qquad (X,Y\ \text{homogeneous})
\label{eq:IT-15}
\end{equation}
defines a Lie superalgebra structure on \(A_m\). In particular:
\begin{enumerate}
\item \(\deg[X,Y]_s=\deg X+\deg Y\ (\mathrm{mod}\;2)\);
\item \( [X,Y]_s=-(-1)^{\deg X\,\deg Y}[Y,X]_s \) for homogeneous \(X,Y\);
\item the super Jacobi identity holds for all homogeneous \(X,Y,Z\in A_m\):
\begin{multline}
(-1)^{\deg X\,\deg Z}[X,[Y,Z]_s]_s
+
(-1)^{\deg Y\,\deg X}[Y,[Z,X]_s]_s
\\
+
(-1)^{\deg Z\,\deg Y}[Z,[X,Y]_s]_s
=0.
\label{eq:IT-16}
\end{multline}
\end{enumerate}
\end{theorem}

\begin{proof}
Compatibility of \(p\) with multiplication follows from the index addition in \eqref{eq:IT-03}:
\[
p(\alpha\oplus_2\alpha',\beta\oplus_2\beta')
\equiv p(\alpha,\beta)+p(\alpha',\beta')\quad(\mathrm{mod}\;2),
\]
hence \(A_m^{\bar i}A_m^{\bar j}\subseteq A_m^{\overline{i+j}}\).

It is a standard fact that any \(\mathbb Z_2\)-graded associative algebra becomes a Lie superalgebra
under the supercommutator \eqref{eq:IT-15}: bilinearity and the degree rule are immediate; graded
skew-symmetry is a direct rearrangement; and the super Jacobi identity follows by expanding each
bracket into products and canceling terms using associativity and the parity exponents.
\end{proof}


\section{Conclusion}

We have introduced Natural Qubit Algebra (NQA) as a compact real operator
calculus for qubit systems, based on a $2\times2$ real block alphabet
$\{I,X,Z,W\}$ and tensor-word representations indexed by binary codes.
The central principle is that operator structure, rather than Hilbert-space
dimension alone, governs the effective complexity of many quantum constructions.

Our main contributions can be summarized as follows.

\paragraph{(1) Canonical real Clifford organization.}
We showed that the $m$-qubit block basis admits a natural
$(\mathbb Z_2)^{2m}$-grading with a bicharacter controlling commutation signs.
In the two-qubit case, this yields an explicit identification
\[
\mathrm{Mat}(4,\mathbb R)\;\cong\;\mathrm{Cl}(2,2;\mathbb R),
\]
providing a canonical real Clifford normal form for two-qubit circuits.
This makes grading and commutation signatures intrinsic rather than
presentation-dependent.

\paragraph{(2) Bell--CHSH as algebraic non-embeddability.}
Rewriting the Bell--CHSH scenario entirely in the spinor/NQA framework,
we exhibited explicit real matrix representatives of both the quantum
CHSH operator (with spectrum $\{\pm 2\sqrt2,0,0\}$) and its classical
commutative counterpart (spectrum $\{\pm 2\}$).
This allows the violation to be formulated as a precise algebraic statement:
there is no embedding of the noncommutative spinor algebra into a single
commutative Kolmogorov algebra that preserves multiplication and spectral
relations.
In this sense, the Bell bound becomes a structural property of commutative
operator algebras rather than a probabilistic paradox.

\paragraph{(3) Structured oracle representations.}
For both the Bernstein--Vazirani and Grover phase oracles,
we derived compact NQA expressions as products of local block factors.
In particular, the Grover diffusion operator
\[
D = 2|s\rangle\!\langle s| - I
\]
admits a tensor-product description whose symbolic size grows linearly in $m$
when kept in factored form.
Thus, at the level of operator syntax, both ``Clifford'' and
``non-Clifford'' oracle examples can possess comparably compact structured
descriptions.

\paragraph{(4) Spectral boundary of Clifford dynamics.}
We clarified the relation to the Gottesman--Knill theorem by isolating
the spectral mechanism that separates Clifford and non-Clifford behavior.
Clifford operations generate discrete phase spectra confined to
$\{0,\pi/2,\pi,3\pi/2\}$, while the Grover iterate produces
a genuine continuous rotation in a two-dimensional invariant subspace.
NQA does not circumvent the $\Theta(\sqrt N)$ query lower bound;
rather, it makes explicit how non-Clifford dynamics emerge from products
of reflections while retaining compact algebraic form.

\medskip

Taken together, these results suggest that the traditional
``Clifford vs.\ non-Clifford'' dichotomy can be refined by separating
three notions:
\begin{enumerate}
  \item gate-set generation,
  \item spectral phase structure,
  \item symbolic growth of operator expressions.
\end{enumerate}
NQA provides a language in which these aspects can be analyzed independently.

Finally, we emphasize that the present framework is not an isolated
formal construction.
Its intrinsic grading places it naturally within the theory of
color-graded algebras and real Clifford structures.
This connection opens the possibility of importing a substantial
algebraic toolkit into quantum operator analysis. The proposed method can allow working with real quantum systems by studying non-trivial symmetries hidden in them at the level of the mathematical carrier of quantum theory.

\bibliographystyle{unsrt}
\bibliography{references}

@article{Benioff1980,
  author  = {Benioff, Paul},
  title   = {The computer as a physical system: A microscopic quantum mechanical Hamiltonian model of computers as represented by Turing machines},
  journal = {Journal of Statistical Physics},
  volume  = {22},
  number  = {5},
  pages   = {563--591},
  year    = {1980},
  doi     = {10.1007/BF01011339}
}

@book{NielsenChuang2010,
  author    = {Nielsen, Michael A. and Chuang, Isaac L.},
  title     = {Quantum Computation and Quantum Information: 10th Anniversary Edition},
  publisher = {Cambridge University Press},
  year      = {2010},
  isbn      = {9781107002173},
  url       = {https://www.cambridge.org/9781107002173}
}

@article{Feynman1982,
  author  = {Feynman, Richard P.},
  title   = {Simulating physics with computers},
  journal = {International Journal of Theoretical Physics},
  year    = {1982},
  volume  = {21},
  number  = {6-7},
  pages   = {467--488},
  doi     = {10.1007/BF02650179}
}

@book{BengtssonZyczkowski2017,
  author    = {Bengtsson, Ingemar and {\.Z}yczkowski, Karol},
  title     = {Geometry of Quantum States: An Introduction to Quantum Entanglement},
  edition   = {2},
  publisher = {Cambridge University Press},
  year      = {2017},
  doi       = {10.1017/9781139207010},
  isbn      = {9781107026254}
}

@article{VanLoan2000,
  author  = {Van Loan, Charles F.},
  title   = {The ubiquitous {Kronecker} product},
  journal = {Journal of Computational and Applied Mathematics},
  volume  = {123},
  number  = {1--2},
  pages   = {85--100},
  year    = {2000},
  doi     = {10.1016/S0377-0427(00)00393-9}
}

@article{Pauli1927,
  author  = {Pauli, W.},
  title   = {Zur Quantenmechanik des magnetischen Elektrons},
  journal = {Zeitschrift f{\"u}r Physik},
  volume  = {43},
  number  = {9--10},
  pages   = {601--623},
  year    = {1927},
  doi     = {10.1007/BF01397326}
}

@article{Dirac1928,
  author  = {Dirac, P. A. M.},
  title   = {The Quantum Theory of the Electron},
  journal = {Proceedings of the Royal Society of London. Series A},
  volume  = {117},
  number  = {778},
  pages   = {610--624},
  year    = {1928},
  doi     = {10.1098/rspa.1928.0023}
}

@book{Porteous1995,
  author    = {Porteous, Ian R.},
  title     = {Clifford Algebras and the Classical Groups},
  publisher = {Cambridge University Press},
  year      = {1995},
  isbn      = {0521551773},
  doi       = {10.1017/CBO9780511470912}
}

@article{MarkovShi2008,
  author  = {Markov, Igor L. and Shi, Yaoyun},
  title   = {Simulating Quantum Computation by Contracting Tensor Networks},
  journal = {SIAM Journal on Computing},
  volume  = {38},
  number  = {3},
  pages   = {963--981},
  year    = {2008},
  doi     = {10.1137/050644756}
}

@article{GrayKourtis2021,
  author  = {Gray, Johnnie and Kourtis, Stefanos},
  title   = {Hyper-optimized tensor network contraction},
  journal = {Quantum},
  volume  = {5},
  pages   = {410},
  year    = {2021},
  doi     = {10.22331/q-2021-03-15-410}
}

@article{AaronsonGottesman2004,
  author  = {Aaronson, Scott and Gottesman, Daniel},
  title   = {Improved simulation of stabilizer circuits},
  journal = {Physical Review A},
  volume  = {70},
  pages   = {052328},
  year    = {2004},
  doi     = {10.1103/PhysRevA.70.052328}
}

@article{BravyiGosset2016,
  author  = {Bravyi, Sergey and Gosset, David},
  title   = {Improved Classical Simulation of Quantum Circuits Dominated by Clifford Gates},
  journal = {Physical Review Letters},
  volume  = {116},
  pages   = {250501},
  year    = {2016},
  doi     = {10.1103/PhysRevLett.116.250501}
}

@article{BravyiEtAl2019LowRank,
  author  = {Bravyi, Sergey and Browne, Dan E. and Calpin, Padraic and Campbell, Earl and Gosset, David and Howard, Mark},
  title   = {Simulation of quantum circuits by low-rank stabilizer decompositions},
  journal = {Quantum},
  volume  = {3},
  pages   = {181},
  year    = {2019},
  doi     = {10.22331/q-2019-09-02-181}
}

@article{Pashayan2015,
  author  = {Pashayan, Hakop and Wallman, Joel J. and Bartlett, Stephen D.},
  title   = {Estimating Outcome Probabilities of Quantum Circuits Using Quasiprobabilities},
  journal = {Physical Review Letters},
  volume  = {115},
  pages   = {070501},
  year    = {2015},
  doi     = {10.1103/PhysRevLett.115.070501}
}

@article{BravyiSmithSmolin2016,
  author  = {Bravyi, Sergey and Smith, Graeme and Smolin, John A.},
  title   = {Trading Classical and Quantum Computational Resources},
  journal = {Physical Review X},
  volume  = {6},
  pages   = {021043},
  year    = {2016},
  doi     = {10.1103/PhysRevX.6.021043}
}

@book{Lounesto2001,
  author    = {Lounesto, Pertti},
  title     = {Clifford Algebras and Spinors},
  edition   = {2},
  series    = {London Mathematical Society Lecture Note Series},
  volume    = {286},
  publisher = {Cambridge University Press},
  year      = {2001},
  isbn      = {9780521005517},
  doi       = {10.1017/CBO9780511526022},
  url       = {https://www.cambridge.org/core/books/clifford-algebras-and-spinors/8318F7DD5B5DE06B30BC612BB5617021}
}

@misc{Gottesman1998Heisenberg,
  author       = {Gottesman, Daniel},
  title        = {The Heisenberg Representation of Quantum Computers},
  year         = {1998},
  eprint       = {quant-ph/9807006},
  archivePrefix= {arXiv},
  doi          = {10.48550/arXiv.quant-ph/9807006}
}

@article{Vidal2003,
  author  = {Vidal, Guifr{\'e}},
  title   = {Efficient Classical Simulation of Slightly Entangled Quantum Computations},
  journal = {Physical Review Letters},
  volume  = {91},
  pages   = {147902},
  year    = {2003},
  doi     = {10.1103/PhysRevLett.91.147902}
}

@article{BernsteinVazirani1997,
  author  = {Bernstein, Ethan and Vazirani, Umesh V.},
  title   = {Quantum complexity theory},
  journal = {SIAM Journal on Computing},
  volume  = {26},
  number  = {5},
  pages   = {1411--1473},
  year    = {1997},
  doi     = {10.1137/S0097539796300921}
}

@article{BerthiaumeBrassard1994,
  author  = {Berthiaume, Andr{\'e} and Brassard, Gilles},
  title   = {Oracle Quantum Computing},
  journal = {Journal of Modern Optics},
  volume  = {41},
  number  = {12},
  pages   = {2521--2535},
  year    = {1994},
  doi     = {10.1080/09500349414552351}
}

@article{Kac1977,
  author  = {Kac, Victor G.},
  title   = {Lie superalgebras},
  journal = {Advances in Mathematics},
  volume  = {26},
  number  = {1},
  pages   = {8--96},
  year    = {1977},
  doi     = {10.1016/0001-8708(77)90017-2},
  url     = {https://www.sciencedirect.com/science/article/pii/0001870877900172}
}

@book{Scheunert1979LNM,
  author    = {Scheunert, Manfred},
  title     = {The Theory of Lie Superalgebras: An Introduction},
  series    = {Lecture Notes in Mathematics},
  volume    = {716},
  publisher = {Springer},
  address   = {Berlin, Heidelberg},
  year      = {1979},
  doi       = {10.1007/BFb0070929},
  url       = {https://link.springer.com/book/10.1007/BFb0070929},
  isbn      = {978-3-540-09256-8}
}

@article{Scheunert1979GeneralizedLie,
  author  = {Scheunert, Manfred},
  title   = {Generalized Lie algebras},
  journal = {Journal of Mathematical Physics},
  volume  = {20},
  pages   = {712--720},
  year    = {1979},
  doi     = {10.1063/1.524113},
  url     = {https://doi.org/10.1063/1.524113}
}

@article{RittenbergWyler1978GenSuperalgebras,
  author  = {Rittenberg, V. and Wyler, D.},
  title   = {Generalized superalgebras},
  journal = {Nuclear Physics B},
  volume  = {139},
  pages   = {189--202},
  year    = {1978},
  doi     = {10.1016/0550-3213(78)90186-4},
  url     = {https://doi.org/10.1016/0550-3213(78)90186-4}
}

@article{Bruce2019Z2nSupersymmetry,
  author  = {Bruce, Andrew James},
  title   = {On a $\mathbb{Z}_2^n$-Graded Version of Supersymmetry},
  journal = {Symmetry},
  volume  = {11},
  number  = {1},
  pages   = {116},
  year    = {2019},
  doi     = {10.3390/sym11010116},
  url     = {https://www.mdpi.com/2073-8994/11/1/116}
}

@book{Varadarajan2004Supersymmetry,
  author    = {Varadarajan, V. S.},
  title     = {Supersymmetry for Mathematicians: An Introduction},
  series    = {Courant Lecture Notes in Mathematics},
  volume    = {11},
  publisher = {American Mathematical Society},
  year      = {2004},
  isbn      = {978-0-8218-3574-6},
  url       = {https://bookstore.ams.org/cln-11}
}

@article{Tsirelson1980,
  author  = {Tsirelson, Boris S.},
  title   = {Quantum generalizations of Bell's inequality},
  journal = {Letters in Mathematical Physics},
  volume  = {4},
  number  = {2},
  pages   = {93--100},
  year    = {1980},
  doi     = {10.1007/BF00417500}
}

@article{Aspect1982,
  author  = {Aspect, Alain and Dalibard, Jean and Roger, G{\'e}rard},
  title   = {Experimental test of {Bell}'s inequalities using time-varying analyzers},
  journal = {Physical Review Letters},
  volume  = {49},
  number  = {25},
  pages   = {1804--1807},
  year    = {1982},
  doi     = {10.1103/PhysRevLett.49.1804}
}

@article{Hensen2015,
  author  = {Hensen, Bas and Bernien, Hannes and Dr{\'e}au, Ana{\"i}s and Reiserer, Andreas and Kalb, Norbert and Blok, Machiel S. and Ruitenberg, Jeroen and Vermeulen, Raymond F. L. and Schouten, Raymond N. and Abell{\'a}n, Carlos and Amaya, W. and Pruneri, Valerio and Mitchell, Morgan W. and Markham, Matthew and Twitchen, Daniel J. and Elkouss, David and Wehner, Stephanie and Hanson, Ronald},
  title   = {Loophole-free Bell inequality violation using electron spins separated by 1.3 kilometres},
  journal = {Nature},
  volume  = {526},
  number  = {7575},
  pages   = {682--686},
  year    = {2015},
  doi     = {10.1038/nature15759}
}

@article{Giustina2015,
  author  = {Giustina, Marissa and Versteegh, Marijn A. M. and Wengerowsky, Simon and Handsteiner, Johannes and Hochrainer, Armin and Phelan, Kevin and Steinlechner, Fabian and Kofler, Johannes and Larsson, Jan-\AA ke and Abell{\'a}n, Carlos and Amaya, Waldimar and Pruneri, Valerio and Mitchell, Morgan W. and Beyer, J{\"o}rn and Gerrits, Thomas and Lita, Adriana E. and Shalm, Lynden K. and Nam, Sae Woo and Zeilinger, Anton},
  title   = {Significant-loophole-free test of {Bell}'s theorem with entangled photons},
  journal = {Physical Review Letters},
  volume  = {115},
  number  = {25},
  pages   = {250401},
  year    = {2015},
  doi     = {10.1103/PhysRevLett.115.250401}
}

@article{Shalm2015,
  author  = {Shalm, Lynden K. and Meyer-Scott, Evan and Christensen, Bradley G. and Bierhorst, Peter and Wayne, Michael A. and Stevens, Martin J. and Gerrits, Thomas and Glancy, Scott and Hamel, David R. and Allman, Michael S. and Coakley, Kevin J. and Dyer, Sean D. and Hodge, Connor and Lita, Adriana E. and Verma, Varun B. and Lambrocco, Craig and Tortorici, Evan and Migdall, Alan L. and Zhang, Yanbao and Kumor, Daniel R. and Farr, Will and Marsili, Francesco and Shaw, Matthew D. and Stern, Jeffrey A. and Abell{\'a}n, Carlos and Amaya, Waldimar and Pruneri, Valerio and Mitchell, Morgan W. and Knill, Emanuel and Nam, Sae Woo},
  title   = {Strong loophole-free test of local realism},
  journal = {Physical Review Letters},
  volume  = {115},
  number  = {25},
  pages   = {250402},
  year    = {2015},
  doi     = {10.1103/PhysRevLett.115.250402}
}

@article{Fine1982,
  author  = {Fine, Arthur},
  title   = {Hidden variables, joint probability, and the {Bell} inequalities},
  journal = {Physical Review Letters},
  volume  = {48},
  number  = {5},
  pages   = {291--295},
  year    = {1982},
  doi     = {10.1103/PhysRevLett.48.291}
}

@article{DehaeneDeMoor2003Clifford,
  author       = {Dehaene, Jeroen and De Moor, Bart},
  title        = {The Clifford group, stabilizer states, and linear and quadratic operations over GF(2)},
  journal      = {Physical Review A},
  volume       = {68},
  pages        = {042318},
  year         = {2003},
  doi          = {10.1103/PhysRevA.68.042318},
  eprint       = {quant-ph/0304125},
  archivePrefix= {arXiv}
}

@article{Kornyak2011FiniteGroups,
  author       = {Kornyak, Vladimir V.},
  title        = {Computations in finite groups and quantum physics},
  year         = {2011},
  eprint       = {1106.2759},
  archivePrefix= {arXiv},
  primaryClass = {math-ph}
}

@inproceedings{RaptisRaptis2025GlobalMaps,
  author       = {Raptis, Theophanes and Raptis, Vasilios},
  title        = {Symmetries and scale invariance in global maps of quantum circuits},
  booktitle    = {Proceedings},
  volume       = {123},
  number       = {5},
  year         = {2025},
  doi          = {10.3390/proceedings2025123005}
}

@article{Raptis2019BinaryWords,
  author       = {Raptis, Theophanes},
  title        = {Encoding discrete quantum algebras as a hierarchy of binary words},
  journal      = {Journal of Physics: Conference Series},
  volume       = {1251},
  number       = {1},
  pages        = {012041},
  year         = {2019},
  doi          = {10.1088/1742-6596/1251/1/012041}
}

\end{document}